\documentclass[%
a4paper,
superscriptaddress,
notitlepage,
amsmath,amssymb,amsthm,
aps,
10pt,
twocolumn,
prl,
longbibliography,
floatfix
]{revtex4-2}

\usepackage{float}
\usepackage{tikz}
\usetikzlibrary{calc,arrows,decorations.pathreplacing,backgrounds}
\usepackage{tikz}
\usetikzlibrary{fit} 
\usetikzlibrary{shapes.geometric} 
\usetikzlibrary{arrows.meta}
\usetikzlibrary{calc}
\usetikzlibrary{decorations.markings}
\tikzset{font={\fontsize{3pt}{12}\selectfont}}

\colorlet{mylinkcolor}{blue!66!black!80}
\usepackage{mathtools,amssymb,amsfonts}
\usepackage[colorlinks=true,linkcolor=mylinkcolor,citecolor=mylinkcolor,filecolor=cyan,urlcolor=mylinkcolor,breaklinks=true]{hyperref}
\hypersetup{hypertexnames=false}
\usepackage[table]{xcolor}

\usepackage[utf8]{inputenc}
\usepackage{booktabs}
\usepackage{bbm,mathrsfs,mathtools}
\usepackage{graphicx}

\usepackage{booktabs}
\usepackage{diagbox}   
\usepackage{subcaption}

\newcommand{\bra}[1]{\langle#1|}
\newcommand{\ket}[1]{|#1\rangle}

\DeclareSymbolFontAlphabet{\mathscrsfs}{rsfs}

\newcommand{\dS}{\dot{S}}
\newcommand{\dSM}{\dot{S}_{\rm M}}
\newcommand{\hX}{\hat{X}}
\newcommand{\hGamma}{\hat{\Gamma}}
\newcommand{\rP}{{\mathbb P}}
\newcommand{\Ps}{P^{\rm s}}

\usepackage{xparse} 
\usepackage{svg} 
\newcommand{\pinv}[1]{p^{\text{inv}}_{#1}}

\newcommand{\Prob}[1]{\ensuremath{
		{\operatorname{\mathbb{P}}\left[ {#1} \right]}
	}
}

\newcommand{\Oh}[1]{\ensuremath{\mathcal{O}\left(#1\right)}}

\renewcommand{\P}{\mathbb P}

\definecolor{mygray}{HTML}{B2AAB7}

\usepackage{amsthm}

\newtheorem{lemma}{Lemma}
\newtheorem{corollary}{Corollary}

\theoremstyle{plain}
\newtheorem{observation}{Observation}
\newtheorem*{observation*}{Observation}  
\newtheorem{assumption}{Assumption}

\newcommand{\mc}{\mathcal}

\usepackage{physics}

\usepackage{dsfont}

\NewDocumentCommand{\1}{g}{%
	\ensuremath{\mathds{1}%
		\IfNoValueTF{#1}
		{}
		{_{#1}}
	}%
}

\RenewDocumentCommand{\Pr}{g}{%
	\ensuremath{\operatorname{\mathbb{P}}%
		\IfNoValueTF{#1}
		{}
		{\left[ #1 \right]}
	}%
}

\renewcommand{\d}{{\rm d}}

\NewDocumentCommand{\E}{m g}{%
	\ensuremath{\operatorname{\mathbb{E}}%
		\IfNoValueTF{#2}
		{}
		{_{#1}}
		\left[%
		\IfNoValueTF{#2}
		{#1}
		{#2}
		\right]%
	}%
}

\NewDocumentCommand{\Var}{m g}{%
	\ensuremath{\operatorname{Var}%
		\IfNoValueTF{#2}
		{}
		{_{#1}}
		\left[%
		\IfNoValueTF{#2}
		{#1}
		{#2}
		\right]%
	}%
}

\newtheorem{theorem}{Theorem}
\usepackage{relsize}

\newcommand{\func}[1]{{\ensuremath{\mathsf{#1}}}}
\newcommand{\poly}{\func{poly}}

\newcommand{\R}{\mathbb{R}}

\newcommand{\N}{\mathbb{N}}

\newcommand{\tridiagmatrix}[3]{%
	\begin{pmatrix}
		#1 & #2    &        &        &         \\
		#3 & #1    & #2     &        &         \\
		& \ddots & \ddots & \ddots &         \\
		&        & #3     & #1     & #2      \\
		&        &        & #3     & #1
	\end{pmatrix}
}

\usepackage{ragged2e} 
\makeatletter
\long\def\@makecaption#1#2{%
	\vskip\abovecaptionskip
	\begingroup
	\parindent=0pt
	\justifying
	\noindent \small #1.\ \ignorespaces#2\par
	\endgroup
	\vskip\belowcaptionskip
}
\makeatother

\usepackage{placeins} 

\begin{document}
	\title{Consistent Time Reversal and Reliable and Accurate Inference in the Presence of Memory}
	
	\author{Tassilo Schwarz}
	\affiliation{Mathematical bioPhysics group, Max Planck Institute for Multidisciplinary Sciences, 37077 G\"{o}ttingen, Germany}
	\affiliation{Mathematical Institute, University of Oxford, 
		Oxford, OX2 6GG, United Kingdom}
	
	\author{Anatoly B. Kolomeisky}
	\affiliation{Department of Chemistry, Rice University, Houston, Texas 77005, USA}
	\affiliation{Center for Theoretical Biological Physics, Department of Chemical and Biomolecular Engineering and Department of Physics and Astronomy, Rice University, Houston, Texas 77005, USA}
	
	\author{Alja\v{z} Godec}
	\email{agodec@physik.uni-freiburg.de}
	\affiliation{Mathematical Physics and Stochastic Dynamics, Faculty of Mathematics and Physics, University of Freiburg}
	\affiliation{Mathematical bioPhysics group, Max Planck Institute for Multidisciplinary Sciences, 37077 G\"{o}ttingen, Germany}

	\begin{abstract}
		\noindent 
		Thermodynamic inference from coarse observations remains a key challenge. Memory, in particular 
		correlations between consecutively observed 
		mesostates, blur signatures of irreversibility and must be accounted for in defining physical time-reversal, which remains an open problem. We derive an experimentally accessible $k$-th order estimator for the entropy production rate. Using novel measure-theoretic techniques we prove necessary and sufficient conditions for \emph{guaranteed} lower bounds on the 
		dissipation even in the strongly non-Markovian setting. The proof reveals that estimators saturated in the order unravel the duration of memory that needs to be considered in defining physically consistent time-reversal.~We show 
		that 
		Markovian estimators in the absence of a timescale separation lead to artifacts that convey no physical meaning. Similarly, estimators \emph{not} saturated in the order may
		overestimate the dissipation.
		The necessity of \emph{correctly} accounting for memory in thermodynamic inference from strongly non-Markovian observations underscores the still underappreciated challenges and intricacies in defining and understanding 
		irreversibility in
		presence of memory. Our results     
		will 
		hopefully stimulate experiments systematically considering thermodynamic inference on multiple scales \emph{consistently} accounting for memory.
	\end{abstract}
\date[]{Published version: \textit{Physical Review Letters},
	DOI: \href{https://doi.org/10.1103/gwf7-6kvr}{10.1103/gwf7-6kvr}}
	\maketitle

	Measurements in complex systems usually have finite resolution~\cite{Gladrow,Battle,Dieball_PRL,Dieball_PRR} or do \emph{not} resolve all relevant degrees of freedom---they just probe some low-dimensional projections of the dynamics, where many microscopic states become ``lumped'' onto the same observable ``state''. Examples are macroscopic observables in condensed matter systems, such as magnetization~\cite{Ocio}, dielectric response~\cite{dielectric}, and diverse order parameters~\cite{Satya_1,Eli_p,Eli_x}, as well as other observables (e.g.\ molecular extensions or FRET efficiencies and lifetimes) in neural activity~\cite{grandpreDirectEstimatesIrreversibility2024}, single-molecule~\cite{order,Woly,Xie,Dima,Brujic,Hagen,Igor,Olivier,Dima_test,Toolbox} and particle-tracking~\cite{Krueger,Narinder} experiments.

	It is well known that such coarse graining introduces memory in the dynamics of the observable quantities~\cite{Zwanzig,Haenggi_1,Grigolini,rahavFluctuationRelationsCoarsegraining2007}
	(see~\cite{Dieball_PRL} for the effect of a finite resolution). Much less is understood about 
	out-of-equilibrium systems that display memory. In particular, a consistent formulation~\cite{wangDetailedBalanceReversibility2007a,Mehl,Massi,Puglisi,Teza,Talkner,hartichEmergentMemoryKinetic2021a,David,van2022thermodynamic,rahavFluctuationRelationsCoarsegraining2007,dieballPerspectiveTimeIrreversibility2025} and inference~\cite{van2022thermodynamic,Polettini,Challenge,Snippets,Blom_2024,Udo_PNAS} of dissipation in the presence of slow dissipative hidden degrees of freedom is often a daunting task.~It is very difficult to reliably infer violations of time reversibility (i.e., the detailed balance) from experiments on individual molecules~\cite{Challenge,ratzke2014four,dieballPerspectiveTimeIrreversibility2025}.         
	
	Violations of detailed balance on the level of individual stochastic trajectories $\Gamma_t\equiv(x_\tau)_{0\le\tau\le t}$ may be quantified via the steady-state dissipation (entropy production) rate $\dS$. Given a \emph{physically consistent} time-reversal operation $\theta$~accounting for the highly nontrivial (anti-)persistence due to memory~\cite{dieballPerspectiveTimeIrreversibility2025},   $\dS$ is defined as the relative entropy between the probability measure of a path $\rP[\Gamma_t]$ and its time reverse $\rP[\theta\Gamma_t]$ per unit time 
	\cite{Spohn,Gaspard,seifertStochasticThermodynamicsFluctuation2012}, i.e.\
	\begin{equation}
		\dS[\Gamma_t]=k_{\rm B}\lim_{t\to\infty}\frac{1}{t}\left\langle
		\ln\frac{\mathrm d \Pr{\Gamma_t}}{\mathrm d \Pr{\theta\Gamma_t}}\right\rangle,
		\label{S}  
	\end{equation}
	where $\langle\cdot\rangle$ denotes the average over
	${\rm d}\rP[\Gamma_t]$ and $k_{\rm B}$ is Boltzmann's constant. When $x_t$ is an overdamped Markov dynamics, we simply have $\theta\Gamma_t=(x_{t-\tau})_{0\le\tau\le
		t}$.
	Specifically, for an ergodic Markov process on a discrete state space $V$ with transition rates $L_{ji}$ from state $i$ to $j$ and stationary probabilities $\Ps$,
	Eq.~\eqref{S} becomes \cite{Spohn,Gaspard,seifertStochasticThermodynamicsFluctuation2012}
	\begin{equation}
		\dSM[\Gamma_t]=k_{\rm B}\sum_{i<j}(L_{ji}\Ps_i-L_{ij}\Ps_j)\ln\frac{L_{ji}}{L_{ij}}.
		\label{SJ}  
	\end{equation}
	For thermodynamic consistency we must always have ${{L_{ji}>0}\implies {L_{ij}>0}}$ for all ${i,j\ne i\in V}$\cite{seifertStochasticThermodynamicsFluctuation2012}. We henceforth express all estimates of $\dS$ in units of $k_{\rm B}$.

	When a coarse-grained version $\hX_t$ of the process is observed, where several microscopic states are lumped into the same observable state, provided a physically consistent time reversal operation $\theta$ \cite{wangDetailedBalanceReversibility2007a,hartichEmergentMemoryKinetic2021a,hartichCommentInferringBroken2022,Blom_2024}, the definition in Eq.~\eqref{S} with lumped paths $\hGamma_t=(\hX_\tau)_{0\le\tau\le t}$
	may still hold, whereas Eq.~\eqref{SJ} does \emph{not}.
	The reason is that $\hX_t$ displays memory due to coarse-graining, i.e., correlations between consecutively observed mesostates, which must be taken explicitly into account in the time-reversal
	\cite{martinezInferringBrokenDetailed2019,hartichEmergentMemoryKinetic2021a,David,van2022thermodynamic,Blom_2024,Zhao}.
	Determining the precise manner in which memory must enter the consistent time reversal remains elusive. 
	If one simply replaces the rate $L_{ji}$ between microscopic states in the Markovian result~\eqref{SJ} 
	with the effective rates between (mesoscopic) lumped states 
	averaged over $\Ps$, one finds~\cite{Massi} ${\dSM[\hGamma_t]\le\dSM[\Gamma_t]}$.
	However, unless one assumes the existence of a (potentially unphysical~\cite{David}) infinite timescale
	separation between dynamics within and between lumped states,
	$\dSM[\hGamma_t]$ is not the physical dissipation. In particular,
	this inequality does \emph{not} imply that a dependence of
	$\dSM[\hGamma_t]$ on the properties of the lumped state space carries any physical meaning since $\dSM[\hGamma_t]$ is \emph{not} the dissipation rate of the coarse-grained process, ${\dSM[\hGamma_t]\ne
		\dS[\hGamma_t]}$.~In absence of an infinite timescale separation
	thermodynamic interpretations of $\dSM[\hGamma_t]$ (see, e.g.\ 
	\cite{yuInversePowerLaw2021,yuStatespaceRenormalizationGroup2022}) may lead to misconceptions and 
	guaranteed higher-order estimators $\dS^{\rm est}_{k}$ are required, as we explicitly show below.~A time-scale separation on some scale does \emph{not} imply a 
	separation on coarser scales.~Under time-scale separation on \emph{all} scales, $\dSM$ does \emph{not} display any scale dependence, see Appendix~C.

	\begin{figure}[ht!]
		\centering
		\includegraphics[width=0.45\textwidth]{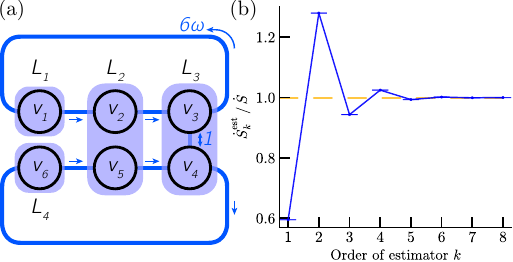}
		\caption{
			(a)~6-state Markov system lumped to a process of order $n>3$. The order is at least 3 because a transition out of $L_3$ depends on whether we arrived from $L_1$ or $L_4$ two state-changes before. Since multiple jumps between $L_2 \leftrightarrow L_3$ may have occurred before, the order is in fact higher, but the influence of back-and-forth jumps vanishes exponentially.~(b)~Entropy production rate (with error bars from 30 i.i.d.\ simulations) estimated with $\dS^{\rm  est}_k$ from trajectories as a function of $k$; saturation is observed from $k=5$ onward. Each depicted
			$\dS^{\rm  est}_k$ was evaluated from path data with $10^9$ transitions with our code optimized for GPUs (available at \cite{GH}) within
			seconds (i.e.\ 1.2 seconds for order $k=1$, and 2.4 seconds for $k=8$ on an H100 GPU). For details see~\cite{SM}.
		}
		\label{fig:cg_and_plot_six_state_main_fig}
	\end{figure}

	Energy dissipation is a critically important quantity since it measures how the system deviates from the equilibrium. It also sets the thermodynamic cost of 
	biological functions~\cite{Qian,Astumian} and transport phenomena~\cite{CSL_3,CSL_4,CSL_5,Saito,Massi-dissipation,VanVu2023PRX,Dieball_2024}, such as molecular-motor mediated
	transport of cargo~\cite{Tolya,Udo},
	biochemical
	oscillations~\cite{Barato,Tu_osci}, error correction in biochemical reactions~\cite{Hopfield,Ninio,mallory2020we}, or cellular computations
	~\cite{comput}.~These functions typically involve
	many processes coupled over multiple spatial and temporal scales
	\cite{Qian,Astumian,yuInversePowerLaw2021}.~It is thus of great interest to unravel how much energy is dissipated on
	different length scales.\\
	
	\indent
	
	Physically inconsistent time reversal ignoring memory may lead to  misinterpretations \cite{yuInversePowerLaw2021}. More importantly, memory implies correlations between consecutively visited mesostates but the extent and manner in which these need to be accounted for when defining time-reversal symmetry, and in turn when estimating the entropy production rate, remain a grand challenge.
	
	We prove here with 
	measure-theoretic techniques that the \emph{only} relevant transition probabilities that enter the estimation of $\dS$ are those accounting for the \emph{relevant extent of memory} (but not those accounting for shorter or longer memory). This means that saturation in the order is an inherent quantifier for the extent of memory that enters the definition of physical time reversal. 
	We construct an experimentally accessible estimator of $\dS$ for strongly non-Markovian observations (see Examples in Figs.~\ref{fig:cg_and_plot_six_state_main_fig},\ref{fig:double-cycle-and-estimates}) and prove that it yields \emph{guaranteed} sharp lower bounds.~Further examples, with strong memory and hidden cycles, demonstrate the power and necessity of the framework,  by highlighting pitfalls in using inconsistent Markovian estimators. 
	
	In particular, we present fractal networks dissipating (i) only on a single (i.e., largest) length-scale and (ii) 
	networks dissipating mostly on the smallest length
	scale and less on larger scales (i.e., networks with an ``inverse energy dissipation cascade'' \cite{yuInversePowerLaw2021}) 
	whereby both display the same scaling of $\dSM$.~A self-similar network structure is shown to \emph{not} 
	imply a power-law scaling of $\dSM$, therefore there is generally \emph{no} relation between the power-law scaling of $\dSM$ and the energy dissipated on a given scale. 
	The apparent power-law dependence of $\dSM$ 
	has \emph{no} relation to the actual energy dissipated on a given scale and may simply be an artifact of an inconsistent time reversal, i.e.\ of neglecting correlations between consecutively visited mesostates due to memory. 
	Most importantly, when no dissipative cycles become hidden by the coarse~graining, any scale dependence vanishes entirely once the memory is accounted for exactly when estimating $\dS$.~Furthermore, with a measure-theoretic framework, our higher-order 
	estimates are proved to be reliable indicators of 
	hidden dissipative cycles and, if saturated in the order that ensures consistent time-reversal, provide \emph{guaranteed} sharper lower bounds. These results emphasize the importance of properly accounting for memory
	in analyzing measurements from coarse observations.
	
	\indent \emph{Setup.---}Consider a continuous-time Markov process $x_t$ on 
	a finite discrete state-space, with $L_{ji}$ denoting transition rates ${i\to j}$. We are interested in the steady-state dissipation rate $\dS$ in Eq.~\eqref{S} in situations where one might only observe some coarse
	(i.e.\ lumped) version of the process that we denote as $\hX_t$. For simplicity, and to align with the assumptions in \cite{yuInversePowerLaw2021}, \emph{we assume the absence of kinetic hysteresis}
	\cite{hartichEmergentMemoryKinetic2021a,David,Blom_2024}, i.e., that the coarse~graining and time reversal commute. The observed process is hence a non-Markovian jump process on a partitioned ``lumped space'', whereby the memory duration is finite but may become arbitrarily long depending on the coarse~graining. In the following examples, the lumped process will correspond to a $k$-th order semi-Markov process with the
	instantaneous state $\hX_t$ depending on the previous $k$ states \cite{zhaoMarkovstateHolography2025}. Thus, $k=1$ corresponds to a renewal process and, if the waiting time is exponentially distributed, to a Markov process.

	\emph{Guaranteed higher-order estimators.---}We will estimate $\dS[\hGamma_t]$ within the simplest Markov approximation as in \cite{yuInversePowerLaw2021} as well as in terms of $k$-th order estimators \cite{Blom_2024} that account for memory.  For simplicity and without any loss of generality
	we will neglect any potential waiting-time contributions (for further details see e.g.~\cite{Blom_2024}). If we let $\hat{\gamma}^{(k_i:k_j)}_k$ denote a subsequence $k_i,k_{i+1},\ldots k_{j}$ of a particular tuple of $k+1$ (with $k\geq 1$) consecutive observed
	states $\hat \gamma_k$ with final state $\hat{\gamma}_k^{({\rm f})}$, and $\mathbb{P}(\hat{\gamma}^{(k_i:k_j)}_k)$ the probability to observe this subsequence along a (formally infinitely long) trajectory, then  the $k$-th order
	estimator of $\dS$ is given by
	\begin{align}
		\dS^{\rm est}_k\equiv
		\frac{1}{T}\sum_{\hat{\gamma}_k}\rP[\hat{\gamma}^{(1:k+1)}_k]\ln\frac{\rP[\hat{\gamma}_k^{({\rm f})} |\hat{\gamma}_k^{(1:k)}]}{\rP[(\theta\hat{\gamma}_k)^{({\rm f})} | (\theta\hat{\gamma}_k)^{(1:k)}]},
		\label{estimator}
	\end{align}  
	where $T$ is the average waiting time per state in the observed trajectories, $\theta\hat{\gamma}_k$ is the corresponding reversed tuple of $k+1$ states, and the upper indices refer to the discrete state sequence \footnote{That is, the sequence of states of the embedded Markov chain.}.~In \cite{SM}, we prove that $\dS^{\rm uncond}_{k} = \sum_{i=1}^k \dS^{\rm est}_{i}$.
	If $\hGamma_t$ is an $n$-th order semi-Markov process then $\dS^{\rm est}_{k}=\dS^{\rm est}_{n},\,\forall k\geq n$ (see proof in~\cite{SM}), and if $\hGamma_t=\Gamma_t$ is a Markov process we have  $\dS^{\rm est}_{1}=\dS$. 
	Notably, if the observed semi-Markov process $\hat \Gamma$ is of order $n$, then 
	$\dS^{\rm est}_k$ for $k\ge n$ provides a lower bound on the microscopic $\dS$
	\begin{align}
		\dS(\Gamma) \geq \dS^{\rm est}_k(\hat \Gamma) \qquad \forall k \geq n.
		\label{ineQQ}
	\end{align}
	We further prove in \cite{SM} that the estimator will reach a plateau upon increasing $k$ once reaching the correct order $n$; at this order (and \emph{not} before) consistent time-reversal is obtained, correctly accounting for correlations between consecutive mesostates. Furthermore, $\dS^{\rm est}_{n}$ (but \emph{not} $\dS^{\rm est}_{k<n}$) always 
	provides a lower bound on $\dS(\Gamma)$ (see our proof in~\cite{SM}; this has been stated, but not proven, for \emph{discrete} time in \cite{roldanEstimatingDissipationSingle2010}). 
	Continuous-time mesoscopic paths require more care since they involve long-time limits of  unbounded hidden-jump sums (see Appendix~D for a trivial argument),
	creating the need for a novel measure-theoretic technique that we develop here. Our proof includes the discrete time result.~Crucially,  \emph{no} statement can be made about the monotonicity of $\dS^{\rm est}_{k}$ as a function of $k$ for $k < n$ as the inequality~\eqref{ineQQ} may be violated (see
	Fig.~\ref{fig:cg_and_plot_six_state_main_fig} and Appendix~A). 
	Inferring the exact 
	$\dS$ from projected observations is generally \emph{not} possible because of dissipative cycles hidden in lumps.
	
	We provide an instructive demonstration of the
	consequences of our framework in Fig.~\ref{fig:cg_and_plot_six_state_main_fig}. The six-state example shows that estimating with insufficient order may lead to underestimation (e.g.\ if a Markovian estimator is used) and even overestimation. Importantly, at sufficiently high order the estimator plateaus. In the present example, the lower bound (Eq.~\eqref{ineQQ}) even becomes sharp.

	\begin{figure}[ht!]
		\centering
		\begin{subfigure}{0.47\linewidth}
			\centering
			\includegraphics[width=\linewidth]{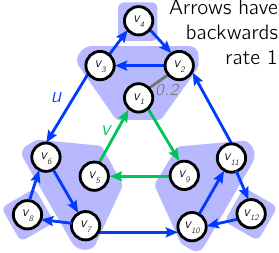}
		\end{subfigure}
		\hfill
		\begin{tabular}[b]{@{}lc@{}}
			\toprule
			estimator & $\dS^{\rm estimated}/ \dS$ \\
			\midrule
			$\dS^{\mathrm{est}}_{1}$ \cite{yuInversePowerLaw2021} & 0.03 \\
			$\dS^{\mathrm{est}}_{2}$ & 0.21 \\
			$\dS^{\rm uncond}_{2}/2$ \cite{yuDissipationLimitedResolutions2024} & 0.12 \\
			$\dS^{\mathrm{est}}_{3}$ &  \cellcolor{mygray} 0.39 \\
			$\dS^{\mathrm{est}}_{4}$ & \cellcolor{mygray}  0.41 \\
			$\dS^{\mathrm{est}}_{5}$ & \cellcolor{mygray}  0.41 \\
			$\dS^{\mathrm{est}}_{6}$ & \cellcolor{mygray}  0.40 \\
			\bottomrule
		\end{tabular}
		\caption{Coarse graining (left) and estimators (right) for $u=v=2$: $\dS^{\mathrm{est}}_{k}$ allows for good inference of $\dS$ already for $k=2$ and saturates in the order for $k\geq 3$ up to statistical fluctuations, indicated in gray.
		}
		\label{fig:double-cycle-and-estimates}
	\end{figure}

	\emph{Sharpening
		inference at strong memory.---}The system in Fig.~\ref{fig:double-cycle-and-estimates} illustrates the benefits of accounting for long-range memory in thermodynamic inference: We observe that accounting for progressively longer memory $k$ yields a better estimate of $\dS$ via $\dS^{\mathrm{est}}_{k}$, until we reach saturation for $k \geq 3$.
	Therefore, physically consistent time-reversal requires accounting for $k=3$ past states.
	Note the substantial improvement over the Markovian estimate $\dSM = \dS^{\rm est}_1$ \cite{yuInversePowerLaw2021}. Further, while   $\dS^{\rm uncond}_{2}/2$
	provides a lower bound~\cite{yuDissipationLimitedResolutions2024}  (this estimator was originally developed for the broader setting of blurred transitions \cite{ertelEstimatorEntropyProduction2024}), we recover with $\dS^{\mathrm{est}}_{3}$ in Fig.~\ref{fig:double-cycle-and-estimates} substantially higher fraction of the microscopic dissipation, and get saturation for $k\geq 3$. In contrast to Fig.~\ref{fig:cg_and_plot_six_state_main_fig}, full recovery is not feasible due to indistinguishable hidden cycles.
	~Further parameters 
	are considered in Appendix~E.
	An efficient GPU implementation of the estimator is provided in \cite{GH}.
	We now turn to the necessity of accounting for memory in time reversal, when considering sequential (scale-dependent) lumping.

	\begin{figure*}
		\centering
		\includegraphics[width=0.85\textwidth]{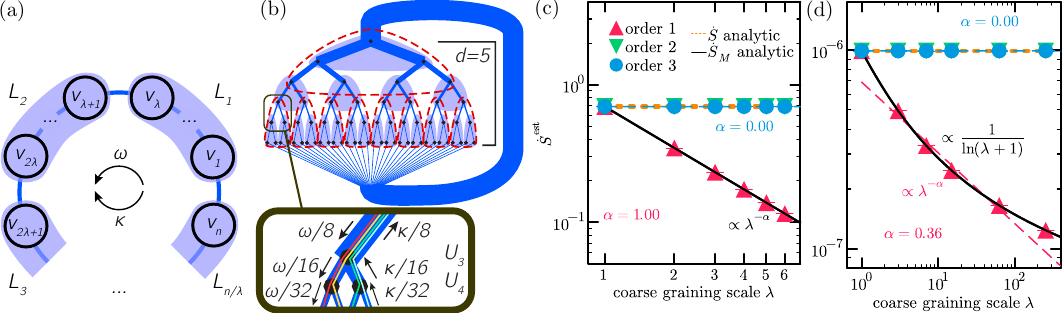}
		\caption{Model systems with sequential lumping and their entropy production rate estimation:~(a) Cycle graph with transition rates $\omega$ and $\kappa$ in counter-clockwise ($+$) and clockwise ($-$) direction, respectively. Lumps $L_1, L_2, \ldots, L_{n/\lambda}$ are highlighted in blue. (b)~Tree of depth $d=5$ with lumps of depth $l=1$ (shaded in blue) and depth $l=2$ (circled in dashed red).~Transition rates down and up (cyclic with respect to the thick edge) are set  ${\propto \omega}$ and ${\propto \kappa}$, respectively, chosen such that the cumulative transition rate between each level is ${\omega - \kappa}$.~Self-similarity becomes obvious upon noticing that each tree consists of multiple smaller trees.~Inset:~Four fundamental cycles passing through the highlighted subgraph on levels $U_3$, $U_4$ and corresponding transition rates.
			(c)~$\dot S^{\rm est}_k$ for the ring graph with $n = 60$
			states. Higher-order estimators ($k \geq
			2$) allow for correct inference of $\dS$ (green, blue),
			while omitting memory leads to an artefactual power law (analytic black line overlays red dashed fit to order-1 estimates). Analytical results are fully corroborated by simulations.
			(d)~$\dot S^{\rm est}_k$ for the tree of depth $d=23$,  i.e.\ $n = 16'777'215$. Higher-order estimators, in combination with exploiting the system's symmetry, allow for correct
			inference of $\dS$ (green, blue), whereas ignoring memory leads to a spurious inverse energy-dissipation cascade (black line), which may be easily mistaken
			for a power law (red dashed fit). In both cases, we considered $100$ independent stationary trajectories each visiting  $10^9$ microscopic states.}
		\label{fig:cycle-and-tree}
	\end{figure*}

	\emph{On hierarchical lumping  and $\dS$.---}We first consider Markovian
	dynamics on a (trivially self-similar) simple cycle graph with $n\ge 4$ states with forward and backward transition
	rates $\omega$ and $\kappa\ne\omega$, respectively (see Fig.~\ref{fig:cycle-and-tree}a). The system 
	has a single dissipative scale with an exact total dissipation rate
	\begin{align}
		\dS = ( \omega - \kappa) \ln\left(\frac{\omega}{\kappa}  \right).
		\label{S_cycle}
	\end{align}
	As in \cite{yuInversePowerLaw2021}, we decimate the state space in $n/\lambda$ observable mesostates (we assume $n =0\, {\rm  mod}\,\lambda$). The coarse process is non-Markovian, to be precise, it is a $2^{\rm nd}$-order semi-Markov process (see also \cite{David}). This means that the waiting time within, 
	and probability for a forward/backward transition from, a lump depends on the previous state.~The exact $\dS$
	may be determined in terms of $J$ and the two-step affinity $A_2$~\cite{martinezInferringBrokenDetailed2019,David,Blom_2024}
	\begin{align}
		\dS_2^{\rm est}&=JA_2=J\ln\left(\frac{\Phi_{+|+}}{\Phi_{-|-}}\right)=\dS,
		\label{2_cycle}
	\end{align}
	where $\Phi_{\pm|\pm}$ denotes the conditional stationary probability for a step in the $\pm$ direction given that the previous step was also in the $\pm$ direction. The proof of the last equality is given in the SM and verified by simulations in Fig.~\ref{fig:cycle-and-tree}c. Equality, as opposed to a lower bound, emerges because no dissipative cycles become hidden, and the waiting-time contribution vanishes (see SM).
	
	Assuming instead memoryless coarse dynamics with ``Markovian'' forward and backward rates defined as in
	\cite{yuInversePowerLaw2021} given by $\Omega (\lambda) =
	\omega/\lambda$ and ${\rm K}(\lambda) = \kappa/\lambda$, respectively, we obtain a single dissipative cycle with homogeneous steady-state
	current $J=(\lambda/n)[\Omega (\lambda)- {\rm K}(\lambda)] =(\omega -\kappa)/n$ and affinity $A=(n/\lambda)\ln(\omega/\kappa)$, yielding the Markov estimate
	\begin{align}
		\dSM = J A =  \frac{1}{\lambda}  \left(\omega
		-\kappa\right)\ln \left(\frac{\omega}{\kappa}\right)
		\propto \lambda^{-1},
		\label{eq:ring:Smist}
	\end{align}	
	indeed an inverse power law according to \cite{yuInversePowerLaw2021}. However, the microscopic system has a single dissipative scale, the entire cycle. Thus, the scaling has \emph{no} implications for the energy dissipation on distinct length scales; it is merely an artifact of inconsistent time reversal in the presence of memory. Consistent time
	reversal and hence thermodynamics thus improve the inference and remove the artefactual scaling behavior.

	We now consider a perfectly self-similar network with a single dissipative \emph{scale} (i.e. the system size) but multiple (equivalent) dissipative \emph{cycles} each with a stationary current $J_c$~(see Fig.~\ref{fig:cycle-and-tree}b). In this binary tree of depth $d$, every edge is bidirectional with rate $\propto\omega$ in the ``down'' ($+$) and
	$\propto \kappa$ in the ``up'' ($-$) direction, respectively, such that the cumulative transition rate between each level is ${\omega-\kappa}$.  This way, the rates along incident edges in consecutive levels in both
	directions $\pm$ are halved (see inset in ~\ref{fig:cycle-and-tree}b). 
	The detailed mathematical construction is
	described in~\cite{SM}.~The exact entropy production of the system in the cycle basis $\mathcal{C}$ with $2^d$ basis cycles reads (see proof in \cite{SM}) 
	\begin{align}
		\dS = \sum_{c \in \mathcal C} J_c A_c = \frac{d+1}{2^{d+1}-1}(\omega-\kappa)\ln\left(\frac{\omega}{\kappa}\right),
		\label{Tree_exact}
	\end{align}
	where the basis-cycle current (see \cite{schnakenbergNetworkTheoryMicroscopic1976} for details on cycle bases) is $J_c=(\omega -\kappa)/2^d(2^{d+1}-1)$. 
	
	We lump the tree into mesoscopic states each having a depth $l<d$ (see
	Fig.~\ref{fig:cycle-and-tree}b). The lumping yields a $2^{\rm nd}$-order semi-Markov process
	and hides \emph{no} dissipative cycles.
	By accounting for memory via~$\dS_{k\ge 2}^{\rm
		est}$~\cite{martinezInferringBrokenDetailed2019,David,Blom_2024} we
	recover the \emph{exact} microscopic result (see derivation in~\cite{SM})
	\begin{align}
		\dS_2^{\rm
			est}&=\frac{\omega-\kappa}{2^{d+1}-1}\frac{d+1}{l+1}\ln\left(\frac{\Phi_{+|+}}{\Phi_{-|-}}\right)=
		\dS_{k >2}^{\rm est}=\dS,
		\label{2_tree}
	\end{align}
	where $\Phi_{\pm|\pm}$ is the conditional probability for
	a transition between two 
	lumps (at given coarseness) in the $\pm$ direction given that the preceding step occurred in the $\pm$ direction.~The exact result~\eqref{2_tree} is confirmed by simulations~in~Fig.~\ref{fig:cycle-and-tree}d~and generalized to a
	tree-diamond graph in~\cite{SM}.~It is \emph{independent} of $\lambda$ since the microscopic system has a single dissipative scale.
	
	The Markovian estimate $\dSM$ for such a lumped system instead reads (see derivation in~\cite{SM})
	\begin{align}
		\!\!\dSM\!&=\!\frac{d+1}{2^{d+1}-1} \frac{\ln(2)}{\ln(\lambda+1)} (\omega - \kappa)
		\ln\left(\frac{\omega}{\kappa}\right)\!\propto \frac{1}{\ln(\lambda+1)},
		\label{Tree_Markov}
	\end{align}
	where $\lambda=n/n_s= 2^{l+1}-1$.~Now there is in fact
	strictly \emph{no} power-law scaling of $\dSM$ in $\lambda$ despite a perfectly self-similar network structure.~However, one may confuse the scaling for a power-law $\sim \lambda^{-0.36}$~(see~Fig.~\ref{fig:cycle-and-tree}d).~An observed power-law thus has \emph{no}
	relation with energy dissipation on distinct scales and may not even be a power law. It is an artifact of inconsistent time reversal in the presence of memory.~Consistent time reversal 
	removes the artefactual scaling behavior.

	\emph{Conclusion.---}A novel, \emph{experimentally accessible} dissipation estimator for coarse-grained observations in continuous-time was introduced and proved  to yield a guaranteed lower-bound. The estimator is applicable to strongly non-Markovian observations (i.e.\ beyond weakly Markovian \cite{brandnerDynamicsMicroscaleNanoscale2025,brandnerDynamicsMicroscaleNanoscale2025a} and correlation-time expansions \cite{foxCorrelationTimeExpansion1983,foxFunctionalcalculusApproachStochastic1986}).
	For the first time, the estimator provides insight into the extent of memory that needs to be considered when defining physically consistent time-reversal.
	
	The necessity of correctly accounting for memory in time-reversal 
	was demonstrated by invalidating a recently proposed relation between an apparent power-law dependence of Markovian estimates for the energy dissipation rate $\dSM$ from coarse-grained observations, and the actual dissipation on different length scales. We constructed fractal networks with a single dissipative scale and self-similar networks (see Appendix B) with an ``inverse energy dissipation cascade''~\cite{yuInversePowerLaw2021}, where most energy is dissipated on the smallest and less energy on larger scales. Both systems exhibit the same scaling behavior of the energy dissipation at different scales. The apparent scaling exponent gradually reduces upon accounting for memory in the dissipation estimates and only if saturated in the order, provides information on hidden dissipative cycles (see~\cite{SM}); scaling laws ignoring memory reflect a mixture of memory and hidden cycles.~With a hypothesized power-law scaling of entropy production rate~\cite{yuInversePowerLaw2021} being already mentioned but not validated in experimental work~\cite{fosterDissipationEnergyPropagation2023},
	~we hope that our novel framework stimulates experiments systematically considering multiple scales.~When no dissipative
	cycles become hidden by the coarse~graining, any scale dependence vanishes as soon as the memory is accounted for exactly in the time-reversal operation and we can infer the microscopic dissipation. Therefore, an inverse power-law dependence of the dissipation rate generally has \emph{no} implications for scale-dependent energy dissipation and may simply be an artifact of inconsistent time reversal in
	presence of memory.
	
	However, it may still be possible to infer meaningful information about the microscopic dynamics, e.g.\ the presence of hidden dissipative cycles, via higher-order dissipation estimators established here.~These~incorporate a consistent time reversal in the presence of memory and, if (and only if) 
	saturated in the order, yield \emph{experimentally accessible}~\cite{vollmarModelfreeInferenceMemory2024,schornerConformationalMemoryProtein2015}, \emph{feasible}~\cite{songNonMarkovModelsSinglemolecule2023, salnikovUsingHigherorderMarkov2016a}, \emph{sharper} and \emph{guaranteed} lower bounds on dissipation, further highlighting
	the importance and benefits of accounting for memory in analyzing data from coarse observations.

	\emph{Acknowledgments.---}Financial support
	from the European Research Council (ERC) under the European Union’s Horizon Europe research and innovation program (Grant Agreement No. 101086182 to A.~G.) and from the Studienstiftung des Deutschen Volkes, the EPSRC Centre for Doctoral Training in Mathematics of Random Systems: Analysis, Modelling and Simulation (EPSRC Grant No. EP/S023925/1), and the Rhodes Trust (to T. S.) are acknowledged. A.B.K.\ also acknowledges the support from the Welch Foundation (Grant C-1559) and the Center for Theoretical Biological Physics sponsored by the NSF (PHY-2019745).

	\emph{Appendix~A:~Overestimation of $\dS$ with insufficient-order estimators.---}While we proved that the saturated-order estimator always provides a lower bound on the microscopic entropy production rate, this was \emph{not} possible for estimators of insufficient order. It is thus natural to ask whether a too-low-order estimator may lead to an \emph{over}estimation. Indeed, an overestimation occurs in the process shown in Fig.~\ref{fig:cg_and_plot_six_state_main_fig}a: The entropy production rate of the $n>3$ order lumped process is clearly overestimated with $\dS^{\rm  est}_2$.
	In particular, we see from Fig.~\ref{fig:six-state-different-number-steps} that the overestimation persists when increasing the number of steps (from $N_{steps}=10^7$ steps to $N_{steps}=10^9$).\\
	\indent  While we prove that $\dS^{\rm est,uncond}_k$ (see SI) is monotone in $k$, the estimator $\dS^{\rm est}_k$ is in particular \emph{not}, as this example shows. However, it is necessary to use---as defined in Eq.~\ref{estimator}---the splittings rather than the joint probabilities, so that we have robustness to overestimation (see Supplemental Material \cite{SM}).
	\begin{figure}[ht!]
		\centering
		\includegraphics[width=0.3\textwidth]{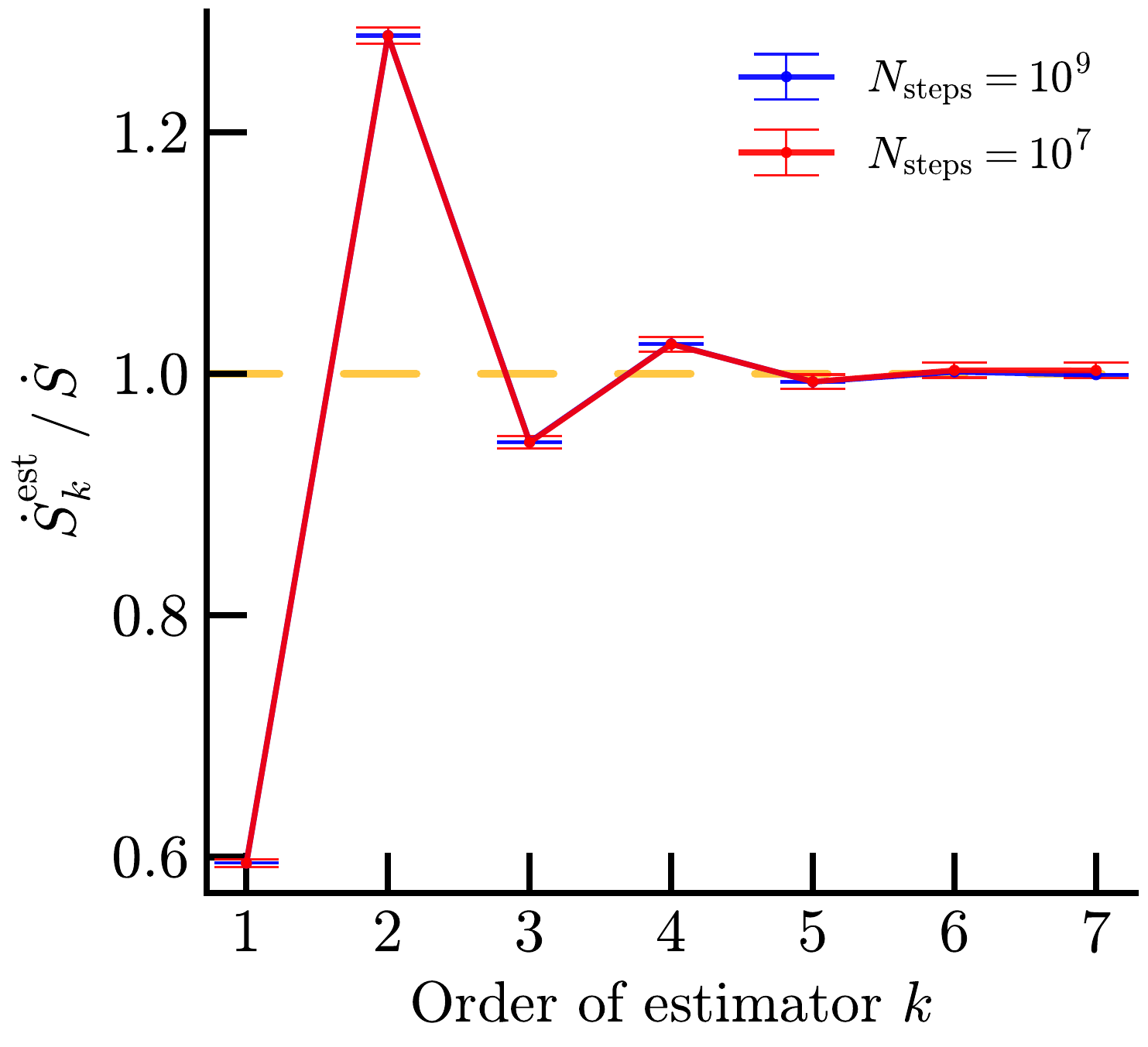} 
		\caption{
			Entropy production rate for the system in Fig.~\ref{fig:cg_and_plot_six_state_main_fig}a estimated with $\dS^{\rm  est}_k$ up to order $k=7$ from lumping of $10^7$ (red) and $10^9$ (blue) microscopic steps. 
			We see that the overestimation at $k=2$ is \emph{not} an undersampling artifact, showing the necessity of sufficiently high estimator order as proved in the present work.
			Each estimator was determined $30$ times; we show the mean and report empirical standard deviations (error bars).
		}
		\label{fig:six-state-different-number-steps}
	\end{figure}

	\emph{Appendix~B~:~Problematic extrapolation to microscopic scales.---}Assuming that the power-law scaling of dissipation rate truly reflects dissipation on distinct length scales, one may attempt to extrapolate
	$\dSM\propto\lambda^{-\alpha}$ and $\dS^{\rm  est}_{k\ge 2}\propto \lambda^{-\alpha}$ to microscopic scales, that is, to vanishing coarse graining $\lambda \to 1$ (see
	\cite{yuInversePowerLaw2021}). In the case of the
	cycle-graph which has a single (i.e.\ macroscopic) dissipative length scale this would coincidentally (and conceptually somewhat paradoxically) yield the
	exact microscopic result (compare Eqs.~\eqref{S_cycle} and
	\eqref{eq:ring:Smist} in the limit $\lambda\to 1$). In all other examples, extrapolating $\dSM$ to $\lambda=1$ would yield a lower bound. This may be useful for bounding $\dS$ from below (the bound $\dSM\le\dS$ was
	obtained in~\cite{Massi}). However, there are
	strictly self-similar networks (e.g.\ the tree and tree-diamond graph) which do \emph{not} display a power-law scaling of energy dissipation rate but where one may nevertheless (erroneously) identify a power law. It is conceivable that such examples may lead to an overestimation of $\dS$. Extrapolated values should thus be interpreted with great care, especially when the 
	underlying microscopic topology of the network is not
	known. Similarly, our examples in Figs.~\ref{fig:cycle-and-tree} and \ref{fig:sierpinski-and-grid} show that one must also \emph{not} extrapolate apparent power-law scalings of higher-order estimates $\dS^{\rm est}_{k\ge 2}$ as these may lead to overestimating $\dS$.
	\begin{figure*}
		\centering
		\includegraphics[width=0.85\textwidth]{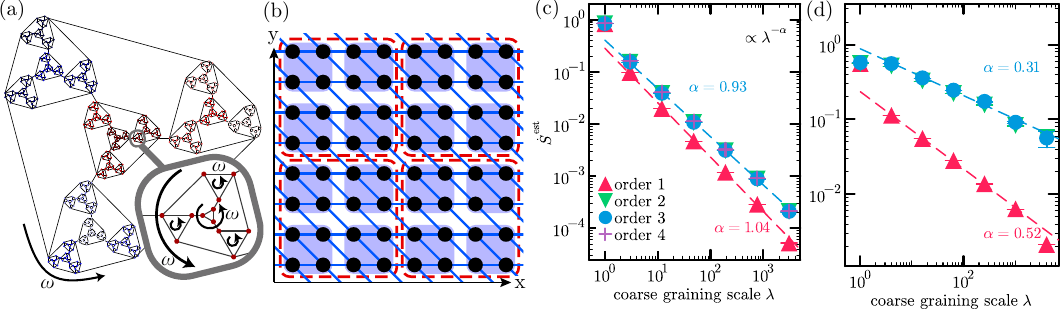} 
		\caption{Model systems and their entropy production rate estimation:~
			~(a)~Self-similar Sierpinski-type graph with lumps of size $\lambda=12$ indicated by the respective vertex coloring.~A stationary current $\omega$ runs through the outer side of each polygon, whereby the size of said polygons depends on the respective recursive depth.~The depicted graph has recursion depth $5$ and $n=768$ vertices.~Inset: Element of recursion-depth $2$ with all irreversible currents indicated by arrows.~(b)~Brusselator depicted as a grid graph as in \cite{yuInversePowerLaw2021}.~A portion
			of the total $n = 450 \times 450$ states is shown with lumps of size
			$\lambda = 2 \times 2$ (shaded in blue) and $\lambda = 4 \times 4$ (circled in dashed red).
			(c)~$\dS^{\rm est}$ for Sierpinski-type graph for recursion depth 7 (${n =
				12288}$ vertices).~We find $\dSM\propto\lambda^{-1}$ while
			higher-order estimators $\dS^{\rm est}_{k\ge 2}$
			yield a power law with ${\alpha=0.93}$, corroborating numerically that the lumped process is a  $2^{\rm
				nd}$ order semi-Markov process.~(d)~$\dS^{\rm est}$ for the
			Brusselator with $n = 202500$ vertices. The virtual power law for
			the Markov estimate ${\dS_1^{\rm est}\propto \lambda^{-0.52}}$ is
			faster decaying than $\dS^{\rm est}_{k\ge 2}\propto\lambda^{-0.31}$
			obtained by accounting for memory. For details, see~\cite{SM}.
		}
		\label{fig:sierpinski-and-grid}
	\end{figure*}

	We conducted numerical experiments to highlight the issue of extrapolating an apparent power law to determine the entropy production rate.
	
	We consider a self-similar graph composed of nested polygons in the spirit of the Sierpinski fractal (see Fig.~\ref{fig:sierpinski-and-grid}a). The rates are set such that each nested polygon is driven with $\omega$. Hence a polygon 
	with recursive depth $l$ has affinity $A_l=3l\ln\omega$ (see derivation in \cite{SM}). The depth $l=1$ corresponds to innermost triangles.~The detailed construction is described in~\cite{SM}. In contrast to the previous examples, this network is designed to indeed display a scale-dependent energy dissipation rate $\dS$, whereby most
	energy is dissipated on the smallest and less energy on larger scales (for quantitative statements see~\cite{SM}). For networks with a recursive depth larger than 3, we are not able to determine $P^{\rm s}_i$ (and hence $\dS$) analytically and therefore resort to numerical methods to compute $\dS$ (see~\cite{SM} for details).  
	The system with recursive depth 7 (12288 microscopic states) was coarse-grained into lumps with recursive depth $l$, such that the
	lumped graph perfectly preserves the original fractal structure (see Fig.~\ref{fig:sierpinski-and-grid}c). However, in stark contrast to the previous examples, the coarse~graining here hides dissipative cycles. We determine $\dSM$
	via Eqs.~(\ref{SJ}-\ref{estimator}) (with $k=1$)  and
	$\dS^{\rm est}_{k\ge2}$ using Eq.~\eqref{estimator} from 108 simulated steady-state trajectories each with a total duration $5\times 10^8$ steps.   
	
	Despite having properties that are ``orthogonal'' to the
	previous examples---featuring a scale-dependent energy
	dissipation rate and dissipative cycles becoming hidden by the lumping---it \emph{resembles} the same power-law scaling of the Markovian estimate $\dSM\propto \lambda^{-\alpha}$ (red line Fig.~\ref{fig:sierpinski-and-grid}c)  in fact with the same exponent $\alpha\approx 1$ as in Eq.~\eqref{eq:ring:Smist}. This shows once more that a
	power-law scaling of $\dSM$, and in particular also the exponent, does \emph{not} give any indication about dissipation on different length scales. Moreover,
	the exponent $\alpha$ changes to $\approx 0.9$  upon accounting for memory via $\dS^{\rm est}_{k\ge 2}$ (blue line Fig.~\ref{fig:sierpinski-and-grid}c), and the lumping yields a $2^{\rm nd}$-order semi-Markov process, since adjacent lumps are connected by exactly one microscopic transition (see proof in~\cite{SM}).
	Accordingly, the numerical results show ${\dS^{\rm est}_{1} \neq \dS^{\rm est}_{2}=\dS^{\rm est}_{3}=\dS^{\rm est}_{4}}$. The fact that $\dS^{\rm est}_{k\ge 2}<\dS$ (i.e.\ equality not reached) is because dissipative cycles are hidden by the coarse graining.

	Finally, let us consider the Brusselator model~\cite{Nicolis,Qian_2002,Brussel_Udo} (see Fig.~\ref{fig:sierpinski-and-grid}b) as in \cite{yuInversePowerLaw2021}, here with 202500
	microscopic states. We perform the coarse graining as was done in \cite{yuInversePowerLaw2021},  
	two levels of coarse~graining are highlighted in Fig.~\ref{fig:sierpinski-and-grid}b.~We estimate $\dS$, $\dSM$, and $\dS^{\rm est}_{k\ge 2}$ from 25 simulated stationary trajectories each
	having $10^9$ steps (skipping $10^9$ initial steps).~Using Eq.~\eqref{estimator} with ${k=1}$ we infer $\dSM$ that agrees with the results in \cite{yuInversePowerLaw2021} (red triangles in
	Fig.~\ref{fig:sierpinski-and-grid}d). In particular, we reproduce the apparent power-law scaling $\dSM\propto\lambda^{-\alpha}$ with the same exponent as in
	\cite{yuInversePowerLaw2021}, i.e.\ $\alpha\approx 0.52$.~Coincidentally, it is of the same order as the \emph{virtual} power-law of the tree and tree-diamond with $\alpha\approx 0.3$ (see Supplemental Material \cite{SM}), which in reality is \emph{not} a power law at all (see Eq.~\eqref{Tree_Markov}).

	We now account for memory in the time-reversal via  $\dS^{\rm est}_{k\ge 2}$.~The lumped process is apparently
	${2^{\rm nd}}$ or ${3^{\rm  rd}}$~order semi-Markov  as  ${\dS^{\rm est}_{2,3}=\dS^{\rm
			est}_{4}}$.~The fact that ${\dS^{\rm est}_{k\ge 2}<\dS}$
	is a consequence of 
	dissipative cycles that become hidden by the coarse~graining.~Higher-order estimates also display an apparent power-law scaling albeit with a smaller exponent ${\dS^{\rm
			est}_{\ge 2}\propto\lambda^{-0.31}}$ for considered $\omega,\kappa$.~However, according to previous
	examples, we must conclude that these apparent power laws
	in fact have \emph{no} implications about energy dissipation on different length scales (see Section~IX.\ in~\cite{SM}).

	\emph{Appendix~C:~Nonexistence of scaling of $\dSM$ under time-scale separation.---}We already showed that a power-law scaling of $\dSM$ has no physical implications in absence of a time-scale separation due to inconsistent time reversal. In the presence of a time-scale separation the time reversal becomes consistent~\cite{Massi}, but the power-law scaling disappears, as we now show. 
	
	Suppose there is a time-scale separation, i.e.\ suppose the dynamics within each element (lump) of a lumping $U$ relaxes fast enough for the transitions \emph{between} elements of $U$ to proceed as in Ref.~\cite{Massi} and (consistently) apply a Markovian estimator $\dS^{\rm est}_1$. 
	We now argue that time-scale separation in $U$ implies time-scale separation in each refined scale. Consider $W$ to be a refinement of $U$ \footnote{That is, $W$ is a lumping that refines the partition $U$.}. Due to time-scale separation in $U$ any transition between and within lumps of $W$ happens at least as fast as in $U$. Hence, we also have time-scale separation in $W$.
	
	Hence, in presence of time-scale separation, the Markovian time reversal is consistent, and the Markovian estimate gives the exact result \cite{Massi} for scale $U$ and any refined scale. This naturally implies the absence of (power law) scaling of $\dS_1^{\mathrm est}$.
	
	\begin{figure}[ht!!]
		\centering
		\includegraphics[width=0.45\textwidth]{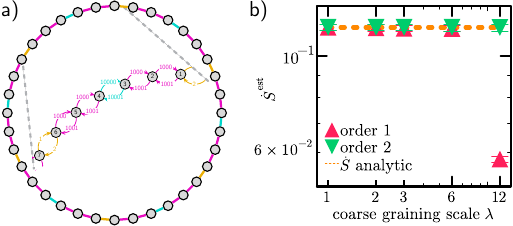}
		\caption{
			(a)~Cycle graph $C_{36}$ with transition rates chosen as indicated in the inset and detailed in \cite{SM} (Section~XIII.) such that the dynamics exhibit time-scale separation.~(b)~Entropy-production rate estimates for $C_{36}$ with time-scale separation between lumps of size $6$: 
			The time-scale separation assumption implies that any refined scale is Markovian ($\dS_1^{est}=\dS_2^{est}$). This also implies, however, that there is \emph{no} scaling of the entropy production rate (observation).
		}  
		\label{Fig6}
	\end{figure}
	
	To illustrate this, consider the cycle graph $C_{36}$ with rates chosen to have time-scale separation (the detailed parameters are given in Fig.~S8 in \cite{SM}).
	At a lumping of size $6$ (first lump are vertices $1, \ldots , 6$, second lump vertices $7, \ldots, 12$ etc.), we observe a time-scale separation. Thus, we may use a Markovian estimator for a lumping of size $6$. To get the correct entropy production rate at a lumping of size $12$, however, we must use an estimator $\dS_k^{\rm est}$ of order $k\geq 2$ (as argued in the letter). As we see in Fig.~\ref{Fig6}, any refinement of the lumping of size $6$ will be Markovian \emph{and result in the same estimate}. Thus, there is no power-law scaling under time-scale separation.
	
	\emph{Appendix~D:~Uncountability of microscopic paths.---}We finally show that the number of microscopic paths projecting onto a mesoscopic path in any 
	time window $[0,t]$ (with $t>0$)  is generally \emph{not} countable. To that end, suppose that a mesostate $M$ contains at least two neighboring microscopic states $\{0,1\}$.  Take some real number $r \in [0,t)$, and consider its binary encoding, say $\operatorname{bin}(r)$. Take the corresponding microscopic trajectory of jumping from $0$ to $1$ at time $r$. This is a unique mapping to a microscopic path, but all those paths correspond to the mesoscopic path of staying in $M$. By uncountability of $\R$, the claim follows.

	\emph{Appendix~E:~Other parameter choices.---}
	We present in Table~\ref{tab:appendix:more-u-v-values} how other choices of $u,v$ in  Fig.~\ref{fig:double-cycle-and-estimates} affect the estimators' quality: Any $k\geq 2$ is better than the Markovian estimate $\dSM= \dS^{\mathrm{est}}_{1}$ \cite{yuInversePowerLaw2021}. We always observe saturation for $k\geq 3$ (up to statistical fluctuations) and our estimate is already at $k=2$ better than $\frac{1}{2}\dS^{\rm uncond}_{2}/\dS$ \cite{yuDissipationLimitedResolutions2024}. Note how saturation happens for $u=1$ already at $k\geq 2$, indicating that at $u=1$ we have an order-2 process: While $\dS^{\mathrm{est}}_{2}$ allows for $100\%$ recovery, using in that setting $\frac{1}{2}\dS^{\rm uncond}_{2}/\dS$ \cite{yuDissipationLimitedResolutions2024} would allow only for recovery of $82\%$ and $83\%$ respectively, with the remaining recovery only possible through the statistically very costly (binning of a continuous-time random variable for integral computation) inference from the waiting time distribution asymmetry. Note that  $\sigma_2^\ell$ in Ref.~\cite{yuDissipationLimitedResolutions2024} is---in contrast to $\sigma_2^t$---statistically equally feasible as $\dS_2^{\rm est}$.
	
	\FloatBarrier
	\begin{table}
		\centering
		
		\begin{tabular}{lccccc}
			\toprule
			\textbf{parameter u} & 2.0 & 1.0 & 2.0 & 1.0 & 1.5\\
			\textbf{parameter v} & 2.0 & 2.0 & 1.0 & 3.0 & 2.0\\
			\midrule
			$\dS^{\mathrm{est}}_{1}/\dS$ \cite{yuInversePowerLaw2021}  & 0.03 & 0.58 & 0.06 & 0.63 & 0.14\\
			$\dS^{\mathrm{est}}_{2}/\dS$   & 0.21 & 1.05 & 0.09 & 1.03 & 0.45\\
			$\frac{1}{2}\dS^{\rm uncond}_{2}/\dS$ \cite{yuDissipationLimitedResolutions2024}  & 0.12 & 0.82 & 0.07 & 0.83 & 0.30\\
			$\dS^{\mathrm{est}}_{3}/\dS$   & \cellcolor{mygray} 0.39 & \cellcolor{mygray} 1.02 & \cellcolor{mygray} 0.20 & \cellcolor{mygray} 1.02 & \cellcolor{mygray} 0.61 \\
			$\dS^{\mathrm{est}}_{4}/\dS$   & \cellcolor{mygray} 0.41 & \cellcolor{mygray} 1.02 & \cellcolor{mygray} 0.22 &\cellcolor{mygray}  1.02 &\cellcolor{mygray}  0.63\\
			$\dS^{\mathrm{est}}_{5}/\dS$   & \cellcolor{mygray} 0.41 &\cellcolor{mygray}  1.00 & \cellcolor{mygray} 0.21 & \cellcolor{mygray} 1.01 &\cellcolor{mygray}  0.62\\
			$\dS^{\mathrm{est}}_{6}/\dS$   &\cellcolor{mygray}  0.4 &\cellcolor{mygray}  0.99 & \cellcolor{mygray} 0.21 &\cellcolor{mygray}  1.00 & \cellcolor{mygray} 0.61\\
			\bottomrule
		\end{tabular}
		
		\caption{Ability of different estimators under more parameter choices for the system of Fig.~\ref{fig:double-cycle-and-estimates}: We observe a consistently better approximation of $\dS$ through $\dS_{k}^{\rm est}$ for every $k\geq 2$ compared to $\dS_M  = \dS_1^{\rm est}$ \cite{yuInversePowerLaw2021} and compared to $\dS_2^{\rm uncond}/2 = \frac{1}{2}(\dS_1^{\rm est} + \dS_2^{\rm est}) = \sigma_2^\ell$ in \cite{yuDissipationLimitedResolutions2024}.}
		\label{tab:appendix:more-u-v-values}
	\end{table}
	\FloatBarrier

\let\oldaddcontentsline\addcontentsline
\renewcommand{\addcontentsline}[3]{}

\let\addcontentsline\oldaddcontentsline

\clearpage
\onecolumngrid

\renewcommand{\thefigure}{S\arabic{figure}}
\renewcommand{\theequation}{S\arabic{equation}}

\setcounter{equation}{0}
\setcounter{figure}{0}
\setcounter{table}{0}
\setcounter{page}{1}
\setcounter{section}{0}

\setcounter{secnumdepth}{3}
\setcounter{tocdepth}{3}

\begin{center}\textbf{Supplemental Material for:\texorpdfstring{\\}{ }Consistent time reversal and reliable and accurate inference in the presence of memory}\\[0.2cm]
	Tassilo Schwarz$^{1,2}$, Anatoly B. Kolomeisky$^{3,4}$, and Alja\v{z} Godec$^{1,5}$\\
	\emph{$^{1}$Mathematical bioPhysics Group, Max Planck Institute for Multidisciplinary Sciences, 37077 G\"ottingen, Germany}\\
	\emph{$^{2}$Mathematical Institute, University of Oxford, 
		Oxford, OX2 6GG, United Kingdom}\\
	\emph{$^{3}$Department of Chemistry, Rice University, Houston,Texas 77005, USA}\\
\emph{$^{4}$Center for Theoretical Biological Physics,
	Department of Chemical and Biomolecular Engineering and
	Department of Physics and Astronomy,
	Rice University, Houston, Texas 77005, USA}
\\
\emph{$^{5}$Mathematical Physics and Stochastic Dynamics,
	Faculty of Mathematics and Physics,
	University of Freiburg}
\\[0.6cm]
\end{center}

\begin{quotation}
	In this Supplemental Material we provide further details on the models, the underlying calculations, and some further technicalities. In addition, we give details on the model parameters and numerical considerations. 
\end{quotation}
\tableofcontents
\newpage
\onecolumngrid      


\section{Notation}
Throughout we use the following conventions:
\begin{itemize}
	\item For $n\in \N$, we define $[n]\equiv \{1, \ldots n\}$.
	\item $L$ corresponds to the generator of the considered underlying/microscopic Markov process; for $x\ne y$, $L_{yx}$ is the rate of the transition $x\to y$.
	\item For a graph $G$, V(G) denotes its vertex set.
	\item We omit $k_B$ and the temperature here for readability.
\end{itemize}

\section{Analysis of higher-order estimators}

In this section, we show first that the defined higher-order estimators (saturated in the order) applied to the observed (lumped) dynamics always provide a lower bound to the true entropy production rate $\dS$ of the \emph{unobserved} (detailed) Markov process.

We then show that using a higher-order estimator beyond the intrinsic (semi-)Markov order of observed process is \emph{not} ``harmful'', i.e.\ that a $(k+n)$-th order estimator for a $k$-th order process yields the same result for any $n \in \N_0$. The higher-order estimators are thus robust against over-estimation.

Throughout, we assume that the microscopic process is a time-homogeneous,
irreducible continuous-time Markov chain on a finite state space $E$,
initialized in its unique stationary distribution.

\subsection{Estimator yields lower bound on microscopic \texorpdfstring{$\dS$}{dS} from observed dynamics}

\label{sec:rigorous-proof}
\begin{theorem}\label{thm:estimator-gives-lower-bound}
	For an observed process $\hat \Gamma$ of (semi-)Markov order $k$, the estimator satisfies
	$$ \dS( \Gamma) \geq \dS_k^{est}(\hat \Gamma),$$
	i.e.\ the estimator recovers from the coarsely observed, lumped dynamics a lower bound on the entropy production rate of the microscopic process.
\end{theorem}\vspace{0.2cm}
The proof of Thm.~\ref{thm:estimator-gives-lower-bound} requires a careful technical treatment. The reason for this is as follows: While the dynamics evolve in continuous time, the history is manifested in the discrete state sequence---e.g., for a $k$-th order process, we need a tuple of $k+1$ observed mesostates, say $\hat \gamma_{k}$. Suppose for now that we fix the time window of observation to $[0,t]$. But the number of microscopic trajectories occurring in $[0,t]$ corresponding to the mesoscopic observation $\hat \gamma_k$, is in general infinite. That is, in general there can be any number of unobserved microscopic trajectories $\gamma$ corresponding to the observation $\gamma_k$. Fixing one such $\gamma$, its length  $m\equiv |\gamma|$  (number of microscopic states) is unbounded. This is because there can be arbitrarily many unobservable microscopic transitions within any observable lumped mesostate.
To obtain the probability that such a microscopic $\gamma$ occurs in time $[0,t]$, we need an $m$-fold convolution of the exponential microscopic waiting times. That is, we have $m$ many integrals to convolve over the time span $[0,t]$.\\
\indent With this in mind, we recall that the entropy production rate (for the microscopic process) is defined as 
\begin{align}
	\dS[\Gamma] \equiv 
	\lim_{t \rightarrow \infty} \frac{1}{t} \E{\mathbb{P}_t}{\ln(\frac{\d \mathbb{P}_t}{\d \mathbb{P}_t^\theta})}.
	\label{def:epr-microscopic}
\end{align}

\paragraph*{Motivation and need for a measure-theoretic approach} 
If we coarse-grain the above from the microscopic dynamics $\Gamma_t$ to the mesoscopic dynamics $\hat \Gamma_t$ we need an expression for the probability of all microscopic $\gamma$ mappings to a particular mesoscopic $\hat \gamma_k$. This is not just an application of the log sum inequality as might be done in the discrete setting; passing the limit $t \rightarrow \infty$ through an unbounded convolution of integrals requires justification.
Further, we must only capture microscopic trajectories of length $t$ in this probability. After finding such an expression, one would have to take the log ratios etc. and divide by $t$. Afterwards, evaluating the limit $t \rightarrow \infty$ of this complicated object would still be required, which is technically very hard. Such an approach was attempted for $k \leq 2$ in ref.~\cite{SUPmartinezInferringBrokenDetailed2019}. However, the mathematical proof therein is not sound (see eqn.~(41) and (50) therein);~One \emph{must not} simply replace the upper boundary in each of the infinitely many integrals entering the convolution with ``$\infty$'' and simply let the time go to infinity without controlling the behavior of the unbounded-fold integral. In particular, when replacing the $t\to\infty$ limit by the number of discrete jumps going to infinity, one has to carefully consider how these quantities interdepend and scale.\vspace{0.2cm}\\
\indent Here, we present a proof of Theorem~\ref{thm:estimator-gives-lower-bound} for any $k$ which is guaranteed to be consistent with Eq.~\eqref{def:epr-microscopic}. To do so, we use a measure-theoretic approach. The key point is that once this framework has been established, we can work directly with measures and \emph{not} with densities which already for fixed $k$ may only be expressed by an infinite number of integrals---an approach that would become extremely hard (at least) when evaluating the limit $t\to\infty$.

\begin{proof}[Proof of Theorem.~\ref{thm:estimator-gives-lower-bound}]
	We first set up the probability space we are working in, then mathematically define "lumping", and finally work in the probability space under lumping.
	\subsubsection{Filtered measure space}
	We consider continuous-time discrete-space Markov Chains.  Let $E$ be the discrete state space of the chain. We make the ("thermodynamic consistence") assumption:
	
	\begin{assumption} \label{assumption:each-transition-reversible}
		A state transition $u \rightarrow v$ has positive probability iff $v \rightarrow u$ does, for any $u,v \in E$.
	\end{assumption}
	Our measure space is 
	\begin{align}
		(\Omega, \mc F, \Pr),
	\end{align}
	where $\Omega \equiv D[0,\infty) \equiv D(E \times [0,t) )$ is the path space of càdlàg functions from $[0,\infty)$ to $E$.
	Let $\xi\equiv \mc P(E)$ be the power set of $E$, such that $(E,\xi)$ is a measurable space.
	The sigma-algebra 
	$\mc F$ is on the path space defined as follows: For time $s\in[0,\infty)$ consider the \emph{evaluation mapping} 
	\begin{align}
		\pi_s: D[0,\infty) &\rightarrow E\\ 
		\pi_s(\omega)&\equiv \omega(s) \quad \text{for a path }\omega \in D[0,\infty),
	\end{align}
	and note that $( \pi_s: s\in [0,\infty))$ is a family of random variables. Then $\mc F$ is the $\sigma$-algebra generated by the family,
	i.e.\ $\mc F$ is the smallest $\sigma$-algebra such that for each $s \in \R_{\geq 0}$, $\pi_s$ is $\xi$-measurable:
	\begin{align}
		\forall l \in \xi: \quad \{\omega \in D[0,\infty): \; \pi_s(\omega) \in l\}\in \mc F.
	\end{align} 
	$\mc F_t$ is the filtration generated from that family of random variables:
	\begin{align}
		\forall l \in \xi: \quad \{\omega \in D[0,\infty): \; \pi_s(\omega) \in l\}\in \mc F_s.
	\end{align}
	
	In the following, we consider the resulting filtered probability space 
	\begin{align}
		(\Omega = D(E \times [0,\infty)), \mc F, (\mc F_t)_{t \geq 0}, \Pr).
	\end{align}
	This induces the family of probability measures $(\Pr_t)_{t \geq 0}$ with 
	\begin{align}
		\Pr_t: \quad \mc F_t &\rightarrow [0,1]\\
		\forall F \in \mc F_t: \quad \Pr_t[F] &\equiv \Pr{F}.
	\end{align}
	
	\subsubsection{Time Reversal}
	
	For a fixed finite window $[0,t]$, we reverse with $\theta_t$ paths using left limits to retain càdlàg paths and note that the map $\theta_t$ acts on the restriction of a path to $[0,t]$; its continuation outside that window is irrelevant to $\mc F_t$. To be precise, we define the reversed measure as the pushforward
	\begin{align}
		\Pr_t^{\theta}(A) &\equiv \Pr_t(\theta_t^{-1} A), \quad A \in \mc F_t\\
		\text{where } (\theta_t \omega)(s)&\equiv \omega((t-s)-) \text{ is the reversed path for } 0 \leq s < t.
	\end{align}

	We define the Radon-Nikodym derivative of $\Pr_t$ with respect to $\Pr_t^\theta$:
	\begin{align}
		X_t(\omega)\equiv\frac{\d \Pr_t}{\d \Pr_t^\theta}\left(\omega \right)
	\end{align} 
	and note that this exists since by Assumption~\ref{assumption:each-transition-reversible} (i.e.\ thermodynamic consistency), we have absolute continuity: $\Pr_t \ll \Pr_t^\theta$ (we even have equivalence, i.e.\ $\Pr_t^\theta \ll \Pr_t$ and hence $\Pr_t^\theta \sim \Pr_t$). By the Radon-Nikodym-Theorem, the function
	\begin{align}
		\frac{\d \Pr_t}{\d \Pr_t^\theta}
	\end{align}
	exists indeed, and is a $\mc F_t$-measurable function $\Omega \rightarrow [0,\infty)$.
	Note that we may have $\Pr_t(\omega)=0$ (and hence $\Pr_t^\theta(\omega)=0$) for some realization $\omega$. This may be avoided (w.l.o.g.) either by limiting the integration     to the support of  $\Pr_t(\omega)$ (and by equivalence to that of $\Pr_t^\theta$) or by setting $\frac{\d \Pr_t}{\d \Pr_t^\theta}(\omega_0)=1$ for all samples $\omega_0$ with $\Pr_t(\omega_0)=0$. For notational convenience we will therefore suppress this point in the notation. 
	Note also that we do not have issues with negativity, since measures are always non-negative, so the Radon-Nikodym derivative is as well always non-negative.

	\subsubsection{Lumping}
	In the lumped setting we cannot observe the full state space $E$, but rather partitions thereof. This means, the measurable space that the evaluation mapping maps onto is only 
	\begin{align}
		(E, \hat \xi).
	\end{align}
	Notably, $E$ remains the same as in the microscopic setting, but we now have a sub-sigma-algebra: $\hat \xi \subset \xi = \mc P(E)$ contains only the partial information we observe. We call the filtration induced by lumping on path space as $\hat{\mc F_t}$. Note that $\hat{\mc F_t} \subseteq \mc F_t$. This shows directly that any $\hat{\mc F_t}$-measurable random variable is also $\mc F_t$-measurable, as desired;~A quantity expressible via the lumped dynamics can in particular be determined if we observe all the microscopic dynamics.
	
	\begin{observation}
		We note that our sigma algebras behave ``nicely'' under time reversal. Since the lumping operation and time-reversal (by the construction of the lumping) commute, we have that both $\Pr$ and $\Pr^\theta$ work on the same sigma-algebra, so in particular on the same measurable space $(D, \mc{\hat F_T} )$, so that $\Pr \ll \Pr^\theta$ is well defined, an the Radon Nikodym derivative is well defined.
	\end{observation}

	\subsubsection{Framework}
	Let us denote the lumps in state space by $S_1, \ldots, S_o$ ($o$ being the number of lumps). These lumps induce the subfiltration $\hat{\mc F_t}$ as described above. We define the index set $\mc S^l$ as all sequences in $[o]^l$ without repetition in subsequent elements
	\begin{align*}
		\mc S^l\equiv \{(i_1, \ldots, i_l) \in [o]^l \mid \forall 1 \leq j <l: i_j \neq i_{j+1} \}.
	\end{align*}
	This ensures that we indeed consider only observable state changes; since we are in continuous-time dynamics, microscopic state changes within the same mesostate (i.e.\ without changing the mesostate) are invisible to the observer and simply prolong the   
	residence time in the mesostate).
	
	We consider the resulting probability space $(\Omega,\mc F, (\mc F_t)_t, \Pr)$ and $(\hat{\mc F_t})_t$.
	
	We further denote by $\ell(x)$ the lump of microstate $x$.
	
	Let $0<\tau_1<\tau_2<\cdots$ be the successive times after $0$ at which $\ell(X_s)$ changes. Set $Z_0:=\ell(X_0)$ and $Z_j:=\ell(X_{\tau_j})$ for $j\geq1$. Define $N_t:=\#\{j:\tau_j\leq t\}$ and $n_t:=N_t+1$. Thus $N_t$ counts observed jumps, while $n_t$ counts successive observed lump visits, including the initial visit.
	
	Note that $ n_t$ is $\hat{\mc F_t}$-measurable. We have (recall the suppressed convention on the support of $\Pr_t$)
	\begin{align}
		&\int_\Omega \ln(\frac{\d \Pr_t}{\d \Pr_t^\theta}) \d \Pr_t \nonumber \\
		=& \int_\Omega \frac{\d \Pr_t}{\d \Pr_t^\theta} \ln(\frac{\d \Pr_t}{\d \Pr_t^\theta}) \d \Pr_t^\theta \tag{by Radon-Nikodym, noting that $\Pr_t \ll \Pr^\theta_t$}. \nonumber\\
		\intertext{Using $Z_0,Z_1,\ldots$, we define for $1\leq l<m$ the events}
		&  A_{i,t}:=\{\omega\in\Omega:N_t(\omega)\geq m-1,\ (Z_0,\ldots,Z_{m-1})=i\}\in\hat{\mc F}_t,\qquad i\in\mc S^m,\nonumber\\
		&  B_{i,t}:=\{\omega\in\Omega:N_t(\omega)=l-1,\ (Z_0,\ldots,Z_{l-1})=i\}\in\hat{\mc F}_t,\qquad i\in\mc S^l,\nonumber
		\intertext{Since for every $m \in \N $ and $t \in \R^+$, the set}
		& \!\!\!\!\!\left\{A_{i,t} |i \in \mc S^m \right\} \bigcup  \left\{B_{i,t} | 1 \leq l < m , i \in \mc S^l \right\}\nonumber
		\intertext{forms a partition of $\Omega$, we have}
		=&\int_\Omega \sum_{i \in \mc S^m} \1_{A_{i,t}} \frac{\d \Pr_t}{\d \Pr_t^\theta} \ln(\frac{\d \Pr_t}{\d \Pr_t^\theta}) \d \Pr_t^\theta + \int_\Omega \sum_{\substack{l < m \\i \in \mc S^l}} \1_{B_{i,t}} \frac{\d \Pr_t}{\d \Pr_t^\theta} \ln(\frac{\d \Pr_t}{\d \Pr_t^\theta}) \d \Pr_t^\theta \nonumber\\
		=&\underbrace{\sum_{i \in \mc S^m} \int_{A_{i,t}} \frac{\d \Pr_t}{\d \Pr_t^\theta} \ln(\frac{\d \Pr_t}{\d \Pr_t^\theta}) \d \Pr_t^\theta}_{(I)} 
		+ \underbrace{\sum_{\substack{l < m \\i \in \mc S^l}} \int_{B_{i,t}} \frac{\d \Pr_t}{\d \Pr_t^\theta} \ln(\frac{\d \Pr_t}{\d \Pr_t^\theta}) \d \Pr_t^\theta}_{(II)}\nonumber
	\end{align}
	
	We treat both terms separately. For $(I)$, we observe that it is of the form $\phi(X_t) = X_t \ln X_t $ for the Radon-Nikodym derivative $X_t \equiv \frac{\d \Pr_t}{\d\Pr_t^\theta}$
	\begin{align}
		(I) =& \sum_{i \in \mc S^m} \int_{A_{i,t}} \frac{\d \Pr_t}{\d \Pr_t^\theta} \ln(\frac{\d \Pr_t}{\d \Pr_t^\theta}) \d \Pr_t^\theta\nonumber \\
		\equiv &  \sum_{i \in \mc S^m} \int_{A_{i,t}} \phi(X_t) \d \Pr_t^\theta\nonumber
		\intertext{ For any measurable event $B$ and event $A$ with positive measure, we define the conditional measure}
		&  \tilde \Pr^\theta_{t,A}(B) \equiv \Pr_t^\theta(B \cap A) / \Pr_t^{\theta}(A),\nonumber
		\intertext{yielding} 
		=&\sum_{i \in \mc S^m} \Pr_t^\theta[A_{i,t}] \int_{A_{i,t}} \phi(X_t) \d \tilde \Pr_{t,A_{i,t}}^\theta \nonumber\\
		\intertext{We are now in a position to employ the \emph{measure theoretic} version of Jensen's inequality in path space. }
		\geq &\sum_{i \in \mc S^m} \Pr_t^\theta[A_{i,t}] \phi\left(\int_{A_{i,t}} X_t \d \tilde \Pr_{t,A_{i,t}}^\theta \right) \tag{Jensen} \\
		=&\sum_{i \in \mc S^m} \Pr_t^\theta[A_{i,t}] \int_{A_{i,t}} \frac{\d \Pr_t}{\d \Pr_t^\theta} \d \tilde \Pr_{t,A_{i,t}}^\theta \ln (\int_{A_{i,t}} \frac{\d \Pr_t}{\d \Pr_t^\theta} \d \tilde \Pr_{t,A_{i,t}}^\theta) \tag{def. of $\phi(\cdot)$ and of $X_t$} \\
		=&\sum_{i \in \mc S^m}
		\Pr_t[A_{i,t}] \ln (\frac{\Pr_t[A_{i,t}]}{\Pr_t^\theta[A_{i,t}] })
		\tag{by Radon-Nikodym, and since $\d \Pr_t \ll \d \Pr_t^\theta$}
	\end{align}

	We next connect our continuous-time description to the discrete jump sequence. For $i= (i_1,\ldots, i_m)$, we define $\Pr^\text{jump}_m[i]$ as the stationary jump-indexed probability of $i$: Select a typical observed jump from a long stationary trajectory, record the lump entered at that jump, and then the following $m-1$ observed jumps.
	Here $Z_0=\ell(X_0)$ is the lump at the fixed clock-time origin and
	need not be a postjump label. The sequence $(Z_1,\ldots,Z_m)$ starts at the
	first observed jump after that origin, whereas
	$\Pr_m^{\mathrm{jump}}$ starts at a typical stationary observed jump.
	Conditional on the same microscopic state entered at their respective
	first jumps, these last two sampling rules give the same law for
	subsequent labels. Their distributions of that microscopic entry state
	can nevertheless differ.
	
	We have for every quantity $h_m(A_{i,t})$
	\begin{align}
		\sum_{i \in \mc S^m}
		\Pr_t[A_{i,t}] \ln (\frac{\Pr_t[A_{i,t}]}{\Pr_t^\theta[A_{i,t}] })
		&= 
		\sum_{i \in \mc S^m} \Pr_t[A_{i,t}] h_m(A_{i,t})  
		+ \sum_{i \in \mc S^m} \Pr_t[A_{i,t}] \ln \frac{\Pr_t[A_{i,t}]}{\Pr_t^\theta[A_{i,t}] \exp(h_m(A_{i,t})) } \notag   \\
		&\geq \sum_{i \in \mc S^m} \Pr_t[A_{i,t}] h_m(A_{i,t})  
		+ \left( \sum_{i \in \mc S^m} \Pr_t[A_{i,t}]\right) \ln \frac{  \sum_{i \in \mc S^m} \Pr_t[A_{i,t}]}{ \sum_{i \in \mc S^m} \Pr_t^\theta[A_{i,t}] \exp(h_m(A_{i,t})) } \tag{log sum inequality on second term} \\
		\intertext{We now define the shorthand $r\equiv m-1$ and the event $C_t\equiv \{N_t \geq r\}$. We further define the time-reversed log-likelihood ratio of sequence $i\in \mc S^r$ as $H_r(i) \equiv \log \frac{\Pr^\text{jump}_r(i)}{\Pr^\text{jump}_r(\overleftarrow i)}$. Due to the first observed state being recorded possibly during its dwell time, we omit the initial label. To that end, we define $h_m(A_{i,t}) \equiv H_r((i_2,\ldots, i_m))$. The corresponding random variable of the trailing $r$-many states $i_2, \ldots, i_m$ is called $W_r$. Noting that the set of $A_{i,t}$ for $i\in \mc S^m$ partitions $C_t$, }
		& = \E{\Pr_t}{\1_{C_t} H_r(W_r)} + \left( \sum_{i \in \mc S^m} \Pr_t[A_{i,t}]\right) \ln \frac{  \sum_{i \in \mc S^m} \Pr_t[A_{i,t}]}{ \sum_{i \in \mc S^m} \Pr_t^\theta[A_{i,t}] \exp(h_m(A_{i,t})) } \label{eq:lower-bound-proof:term-I-partial}
	\end{align}
	We find the following bounds
	
	\begin{align}
		\sum_{i \in \mc S^m} \Pr_t^\theta[A_{i,t}] \exp(h_m(A_{i,t})) 
		&=  \sum_{i \in \mc S^m} \Pr_t^\theta[A_{i,t}] 
		\frac{\Pr^\text{jump}_r((i_2,\ldots, i_m))}{\Pr^\text{jump}_r((i_m, \ldots, i_2))} \notag \\
		&=  \sum_{(i_2, \ldots, i_m) \in \mc S^r} \Pr_t^\theta[(i_2,\ldots, i_m),C_t] 
		\frac{\Pr^\text{jump}_r((i_2,\ldots, i_m))}{\Pr^\text{jump}_r((i_m, \ldots, i_2))} \tag{summing over $i_1$ } \\
		\intertext{Denoting by 
			$\Pr^\theta[(i_2,\ldots, i_m)]$ the probability that the first $r$ observed transitions of the time-reversed process are into $(i_2,\ldots, i_m)$}
		&\leq \sum_{i \in \mc S^r} \Pr^\theta[(i_1,\ldots, i_r)] 
		\frac{\Pr^\text{jump}_r((i_1,\ldots, i_r))}{\Pr^\text{jump}_r((i_r, \ldots, i_1))} \tag{reindexing and inclusion of events} \\
		\intertext{Identifying $\Pr^\text{jump}_r((i_r, \ldots, i_1)) = \Pr^{\theta, \text{jump}}_r((i_1, \ldots, i_r))$, we next close the discrepancy between $\Pr^\theta$ and $ \Pr^{\theta, \text{jump}}$. The two sampling rules give different distributions of the microscopic state entered at the first observed jump. Consequently, their probabilities of later labels can differ even conditional on the same first lump. 
			To that end, let $\eta_x^\theta$ be the probability that the first observed jump after a fixed clock-time origin in the stationary time-reversed process enters microscopic state $x$.
			Let $\rho_x^\theta$ be the probability of a typical mesoscopic jump (in the stationary, reversed process) entering, under the hood, microscopic state $x$. Let $K_x^\theta(i)$ be the probability of $i$ in the reversed process, conditioned on entry into $x$.
			We thus have for every $i\in \mc S^r$}
		& \qquad \Pr^\theta[i] = \sum_x \eta_x^\theta K_x^\theta(i) \leq C_\theta \sum_x \rho_x^\theta K_x^\theta(i) = C_\theta \Pr^{\theta, \notag \text{jump}}_r(i),
		\intertext{where $C_\theta\equiv \max _{x: \rho_x^\theta >0} \frac{\eta_x^\theta}{\rho_x^\theta} < \infty$, since $\eta_x^\theta>0$ implies that a positive-rate observed crossing into $x$ is possible, so that its stationary flux is positive, thus $\rho_x^\theta>0$. Finiteness of $E$ then gives $C_\theta<\infty$.}
		& \leq C_\theta \sum_{i \in \mc S^r}
		\Pr^\text{jump}_r((i_1,\ldots, i_r)) = C_\theta  \notag
	\end{align}
	
	We can thus continue to lower bound Eqn.~\ref{eq:lower-bound-proof:term-I-partial}:
	\begin{align}
		&\geq 
		\E{\Pr_t}{\1_{C_t} H_r(W_r)} + \left( \sum_{i \in \mc S^m} \Pr_t[A_{i,t}]\right) \ln \frac{  \sum_{i \in \mc S^m} \Pr_t[A_{i,t}]}{ C_\theta } 
		\label{eq:lower-bound-proof:term-I-partial2}
	\end{align}
	
	Since the second term on the RHS of Eqn.~\ref{eq:lower-bound-proof:term-I-partial2} is bounded by a constant in $t$, we focus on the first term.
	For $\pinv{x}$ being the steady state probability of state $x$, we have the stationary observed jump rate
	\begin{align}
		\nu = \sum_x \pinv{x} \sum_{y: \ell(y) \neq \ell(x)} L_{yx} \notag 
	\end{align}
	and the resulting mean observed residence time $T = 1/\nu$.

	Importantly, we recall our assumption that the process is of order (at most) $k$. Thus, we have
	
	\begin{align}
		H_r(i) = \log \frac{\Pr^\text{jump}_k((i_1,\ldots, i_k))}{\Pr^{\text{jump}}_k((i_r,\ldots, i_{r-k+1}))} + \sum_{j=1}^{r-k}\log \frac{\P[i_{j+k} \mid (i_j\ldots, i_{j+k-1})]}{\P[i_{j} \mid (i_{j+k}\ldots, i_{j+1})]},   \notag 
	\end{align}
	where the latter are order-$k$ jump probabilities.
	We note that due to finiteness of the set of lumps, there exist for fixed $k$ constants $C_1, C_2$ such that for every $i\in \mc S^k$ which has positive support, $r>k$,and  $1 \leq j \leq r-k$
	\begin{align}
		&\left| \log \frac{\Pr^\text{jump}_k((i_1,\ldots, i_k))}{\Pr^{\text{jump}}_k((i_r,\ldots, i_{r-k+1}))} \right| \leq C_1 \notag \\
		&\left| \log \frac{\P[i_{j+k} \mid (i_j\ldots, i_{j+k-1})]}{\P[i_{j} \mid (i_{j+k}\ldots, i_{j+1})]}  \right| \leq C_2 \notag 
	\end{align}
	
	Due to this boundedness and ergodicity, we obtain 
	\begin{align}
		\frac{1}{r} \E{\Pr}{ H_r(W_r)} \rightarrow \sum_{i \in \mc S^{k+1}}  \Pr^{\text{jump}}_{k+1}(i) \ln \frac{ \Pr[i_{k+1} \mid (i_1, \ldots, i_k)]}{ \Pr[i_{1} \mid (i_{k+1}, \ldots, i_2)]} \notag 
	\end{align}

	Recall that $W_r=(Z_1,\ldots,Z_r)$ is defined on the (infinite) trajectory as the first $r$ labels that are being jumped to. On $C_t$, $\1_{C_t} H_r(W_r)$ depends only on the trajectory upto time $t$. Therefore,
	\begin{align}
		\E{\Pr_t}{\1_{C_t} H_r(W_r)} 
		&=  \E{\Pr}{H_r(W_r)} -  \E{\Pr}{\1_{{C_t}^C} H_r(W_r)} \notag  \\
		&\geq  \E{\Pr}{H_r(W_r)} - (C_1+C_2r) \Pr[C_t^C], \notag 
	\end{align}
	where the last step follows from boundedness  $|H_r(i)| \leq C_1 + C_2 r$.
	Similarly, we can bound from above: $\E{\Pr_t}{\mathbf1_{C_t}H_r(W_r)}\leq\E{\Pr}{H_r(W_r)}+(C_1+C_2r)\Pr(C_t^c)$.

	Taking $r(t) \equiv \lfloor \nu t - (\nu t)^{3/4} \rfloor $
	and accordingly $m(t) \equiv r(t) +1 $,
	we have due to $\E{N_t} = \nu t$ and $\Var{N_t} = \Oh{t}$ ( ~\cite{SUPFeller_68,SUPgodrecheStatisticsOccupationTime2001a})
	by Chebyshev's inequality that $\Pr{C_t^c} = \Pr{N_t < r(t)}  = \Oh{t^{-1/2}}$ 
	Thus, 
	\begin{align}
		\frac{1}{t}  \E{\Pr_t}{\1_{C_t} H_{r(t)}(W_{r(t)})}  \rightarrow  \nu  \sum_{i \in \mc S^{k+1}}  \Pr^{\text{jump}}_{k+1}(i) \ln \frac{ \Pr[i_{k+1} \mid (i_1, \ldots, i_k)]}{ \Pr[i_{1} \mid (i_{k+1}, \ldots, i_2)]} \notag 
	\end{align}

	Further, since $0 \leq \sum_{i \in \mc S^m} \Pr_t[A_{i,t}] = \Pr{C_t} \leq  1$, we can conclude for term $(I)$:
	\begin{align}
		\liminf_{t\rightarrow \infty } \frac{1}{t}  \sum_{i \in \mc S^{m(t)}} \int_{A_{i,t}} \frac{\d \Pr_t}{\d \Pr_t^\theta} \ln(\frac{\d \Pr_t}{\d \Pr_t^\theta}) \d \Pr_t^\theta \geq \nu  \sum_{i \in \mc S^{k+1}}  \Pr^{\text{jump}}_{k+1}(i) \ln  \frac{ \Pr[i_{k+1} \mid (i_1, \ldots, i_k)]}{ \Pr[i_{1} \mid (i_{k+1}, \ldots, i_2)]}= \dS_k^{\rm est} \notag 
	\end{align}

	Since the events $B_{i,t}$, for $1\leq l<m(t)$ and $i\in\mc S^l$,
	partition $C_t^c$, Jensen's inequality followed by the log-sum inequality give
	\begin{align}
		(II)
		&= \sum_{l=1}^{m(t)-1}\sum_{i\in\mc S^l}
		\int_{B_{i,t}} X_t\ln X_t\,\d\Pr_t^\theta \nonumber\\
		&\geq \sum_{l=1}^{m(t)-1}\sum_{i\in\mc S^l}
		\Pr_t(B_{i,t})
		\ln\frac{\Pr_t(B_{i,t})}{\Pr_t^\theta(B_{i,t})} \nonumber\\
		&\geq \Pr_t(C_t^c)
		\ln\frac{\Pr_t(C_t^c)}{\Pr_t^\theta(C_t^c)}
		=0, \notag 
	\end{align}
	where the final equality holds because time reversal preserves the number of
	observed jumps in $[0,t]$. Thus $\theta_t^{-1}(C_t^c)=C_t^c$ and $\Pr_t^\theta(C_t^c) =\Pr_t(C_t^c)$.

	We can thus conclude
	
	\begin{align*}
		\dot S \equiv& \lim_{t\rightarrow \infty} \frac{1}{t}\int_\Omega \ln(\frac{\d \Pr_t}{\d \Pr_t^\theta}) \d \Pr_t  \\
		\geq&  \lim_{t\rightarrow \infty} \frac{(I)_t + (II)_t}{t} \\
		\geq& \dot S_k^\text{est},
	\end{align*}  
	\noindent so that the $k$-th order estimator provides---given that the order of the process is $k$---always a \emph{guaranteed} lower bound.
\end{proof}

\subsection{Robustness of the estimator}

What happens if the estimator we apply to the observed mesoscopic trajectories, $\dS^\text{est}_k$ has a higher order $k > n$ than the memory of the mesoscopic process $n$?

It turns out that the estimator is robust against overestimation, meaning that:\vspace{0.2cm}
\begin{theorem}[Robustness against overestimation] \label{thm:robustness-to-overestimation}
	For an observed process $(\hat \Gamma_t)_{t \in \R^+}$ of (semi-)Markov order $n$, the estimator satisfies
	$$ \dS( \Gamma) \geq \dS_k^{est}(\hat \Gamma) \quad \forall k \geq n.$$
\end{theorem}\vspace{0.2cm}

We will prove Theorem~\ref{thm:robustness-to-overestimation} by first showing\vspace{0.2cm} 
\begin{lemma} \label{lem:estimator-consistent-in-order}
	If $(\hat \Gamma_t)_{t \in \R^+}$ is an $n$-th order semi-Markov process then $\dS^{\rm est}_{k}=\dS^{\rm est}_{n},\,\forall k \geq n$.
\end{lemma}\vspace{0.2cm}

For proving this Lemma it suffices to work in discrete time. This allows us to make the proof accessible without using measure theory. To that end, we use the simpler notation from the Letter (around Eq.~(3) therein) and recall that $\hat \gamma_k$ is a particular tuple of $k+1$ consecutively observed states.
We further introduce notation to work on the process's embedded Markov chain: For a trajectory $\Gamma_t\equiv(x_\tau)_{0\le\tau\le t}$, we denote the random variables of the embedded Markov chain by $X^{(1)},X^{(2)}, \cdots$. For readability purposes, we access the $i$-th embedded Markov chain state of $ \hat \gamma_k$ by $\hat  \gamma_k^{(i)}$ for $i \in [k+1]$ and index a sequence of states by $\hat \gamma_k^{(i:j)} \equiv ( \gamma_k^{(i)}, \gamma_k^{(i+1)}, \ldots, \hat \gamma_k^{(j)})$. If $i>j$, the index sequence is decreasing. We use the same indexing for the random variables.
We further denote by 
\begin{align*}
	\Prob{\hat{\gamma}_{k}^{(k+1)} \mid  \hat{\gamma}_{k}^{(1:k)}} &\equiv \Prob{X^{(k+1)} = \hat{\gamma}_{k}^{(k+1)} \mid X^{(1:k)} = \hat{\gamma}_{k}^{(1:k)}} \\&= \Prob{X^{(k+2)} = \hat{\gamma}_{k}^{(k+1)} \mid X^{(2:k+1)} = \hat{\gamma}_{k}^{(1:k)}},
\end{align*}
where the first equality is a definition, and the second equality demonstrates the shift-invariance, which holds since we are operating in the steady state and since the underlying process is in continuous time (thus, we get no periodicity problems). If we refer to the entire sequence of $k+1$ discrete states of $\hat \gamma_k$, we omit indices.

We note that Lemma~\ref{lem:estimator-consistent-in-order} and Corollary~\ref{cor:s-cond-vs-uncond-recurrence} have been argued for in discrete time systems \cite{SUProldanEstimatingDissipationSingle2010}. We provide a first-principles-based proof and thereby (i) extend to our setting of \emph{continuous-time} Markov Chains and (ii) prove further properties of the estimator. In particular the proof of a lower-bound for continuous-time systems is nontrivial, as should be clear from Section~\ref{sec:rigorous-proof} and Ref.~\cite{SUPmartinezInferringBrokenDetailed2019}.

\begin{proof}[Proof of Lemma~\ref{lem:estimator-consistent-in-order}]
	Suppose $(\hat \Gamma_t)_{t \in \R^+}$  is a $n$-th order semi-Markov process. Let $k \geq n$. By definition, we have
	\begin{align*}
		\dS^{\rm est}_{k}(\hat \Gamma_t)
		&\equiv \frac{1}{T} \sum_{\hat{\gamma}_{k}}\Prob{\hat{\gamma}_{k}}
		\ln\frac{ \Prob{\hat{\gamma}_{k}^{(k+1)} \mid \hat{\gamma}_{k}^{(1:k)}}}{\Prob{\hat{\gamma}_{k}^{(1)} \mid \hat{\gamma}_{k}^{(k+1:2)}}} \tag{by definition} \\
		&= \frac{1}{T}  \sum_{\hat{\gamma}_{k}}\Prob{\hat{\gamma}_{k}}
		\ln\frac{ \Prob{\hat{\gamma}_{k}^{(k+1)} \mid \hat{\gamma}_{k}^{(k-n+1:k)}}}{\Prob{\hat{\gamma}_{k}^{(1)} \mid \hat{\gamma}_{k}^{(n+1:2)}}} \tag{$n$-th order process}\\
		&= \frac{1}{T}  \sum_{\hat{\gamma}_{k}}\Prob{\hat{\gamma}_{k}}
		\ln{ \Prob{\hat{\gamma}_{k}^{(k+1)} \mid \hat{\gamma}_{k}^{(k-n+1:k)}}}
		- \frac{1}{T}  \sum_{\hat{\gamma}_{k}}\Prob{\hat{\gamma}_{k}}
		\ln{ \Prob{\hat{\gamma}_{k}^{(1)} \mid \hat{\gamma}_{k}^{(n+1:2)}}}\\
		&= \frac{1}{T}  \sum_{\hat{\gamma}_{k}^{(k-n+1:k+1)}}\Prob{\hat{\gamma}_{k}^{(k-n+1:k+1)}}
		\ln{ \Prob{\hat{\gamma}_{k}^{(k+1)} \mid \hat{\gamma}_{k}^{(k-n+1:k)}}}\nonumber\\
		&- \frac{1}{T}  \sum_{\hat{\gamma}_{k}^{(1:n+1)}}\Prob{\hat{\gamma}_{k}^{(1:n+1)}}
		\ln{ \Prob{\hat{\gamma}_{k}^{(1)} \mid \hat{\gamma}_{k}^{(n+1:2)}}}   \tag{marginalisation}.  \\  
		\intertext{Using that the embedded Markov chain is in steady state and hence shift invariant, we may do a change of variable:}
		&= \frac{1}{T}  \sum_{\hat{\gamma}_{k}^{(1:n+1)}}\Prob{\hat{\gamma}_{k}^{(1:n+1)}}
		\ln{ \Prob{\hat{\gamma}_{k}^{(n+1)} \mid \hat{\gamma}_{k}^{(1:n)}}}
		- \frac{1}{T}  \sum_{\hat{\gamma}_{k}^{(1:n+1)}}\Prob{\hat{\gamma}_{k}^{(1:n+1)}}
		\ln{ \Prob{\hat{\gamma}_{k}^{(1)} \mid \hat{\gamma}_{k}^{(n+1:2)}}}   \\
		&= \frac{1}{T}  \sum_{\hat{\gamma}_{k}^{(1:n+1)}}\Prob{\hat{\gamma}_{k}^{(1:n+1)}}
		\ln{ \frac{\Prob{\hat{\gamma}_{k}^{(n+1)} \mid \hat{\gamma}_{k}^{(1:n)}}
			}{ \Prob{\hat{\gamma}_{k}^{(1)} \mid \hat{\gamma}_{k}^{(n+1:2)}}
			}
		}  \\
		&= \dS^{\rm est}_{n}(\hat \Gamma_t)
	\end{align*}
\end{proof}

We are now in a position to prove the robustness theorem.

\begin{proof}[Proof of Theorem~\ref{thm:robustness-to-overestimation}]
	For an observed process $(\hat \Gamma_t)_{t \in \R^+}$ of order $n$ and estimator of order $k\geq n$, we have for the microscopic entropy production rate
	\begin{align}
		\dS(\Gamma)
		\geq& \dS^{\rm est}_n(\hat \Gamma_t) \tag{by Theorem~\ref{thm:estimator-gives-lower-bound}} \\
		=&  \dS^{\rm est}_k(\hat \Gamma_t) \tag{by Lemma~\ref{lem:estimator-consistent-in-order}}.
	\end{align}
\end{proof}

\subsection{Positivity and behavior for increasing order}

We recall that the estimator in the main text $\dS^{\rm est}_{k}$ is the expectation of the log \emph{splitting} probability ratios. We here connect this to the expectation of the simple path probability ratios. To that end, we define the $k$-th order \emph{unconditional} entropy production rate estimator
\begin{align}
	\dS^{\rm est, uncond}_{k}\equiv \frac{1}{T}  \sum_{\hat{\gamma}_{k}}\Prob{\hat{\gamma}_{k}}
	\ln\frac{ \Prob{\hat{\gamma}_{k}}}{\Prob{\theta \hat{\gamma}_{k}}}.
\end{align}

The following result shows that, irrespective of the order of the (semi-)Markov process, the unconditional estimator can be expressed as the sum of conditional estimators (that is, as a sum of the estimators with splitting probabilities from the main text). That is\vspace{0.2cm}
\begin{lemma} \label{lem:uncond-vs-cond}
	For all $n>0$:
	$$\dS^{\rm est, uncond}_{n} = \sum_{l=1}^n \dS^{\rm est}_{l}.$$
\end{lemma}

\begin{proof}
	By using 
	$$\Prob{\hat{\gamma}_{k}} 
	= \Prob{\hat{\gamma}_{k}^{(k+1)} \mid \hat{\gamma}_{k}^{(1:k)} }
	\Prob{\hat{\gamma}_{k}^{(k)} \mid \hat{\gamma}_{k}^{(1:k-1)} }
	\cdots 
	\Prob{\hat{\gamma}_{k}^{(2)} \mid \hat{\gamma}_{k}^{(1)} }
	\Prob{\hat{\gamma}_{k}^{(1)}}, $$
	and the corresponding identity for the time reversal, we obtain 
	\begin{align*}
		\dS^{\rm est, uncond}_{k}
		&= \frac{1}{T}  \sum_{\hat{\gamma}_{k}}\Prob{\hat{\gamma}_{k}}
		\ln\frac{ 
			\Prob{\hat{\gamma}_{k}^{(k+1)} \mid \hat{\gamma}_{k}^{(1:k)} }
			\Prob{\hat{\gamma}_{k}^{(k)} \mid \hat{\gamma}_{k}^{(1:k-1)} }
			\cdots 
			\Prob{\hat{\gamma}_{k}^{(2)} \mid \hat{\gamma}_{k}^{(1)} }
			\Prob{\hat{\gamma}_{k}^{(1)}}
		}{ 
			\Prob{\hat{\gamma}_{k}^{(1)} \mid \hat{\gamma}_{k}^{(k+1:2)} }
			\Prob{\hat{\gamma}_{k}^{(2)} \mid \hat{\gamma}_{k}^{(k+1:3)} }
			\cdots 
			\Prob{\hat{\gamma}_{k}^{(k)} \mid \hat{\gamma}_{k}^{(k+1)} }
			\Prob{\hat{\gamma}_{k}^{(k+1)}}
		}\\
		&=\dS^{\rm est}_{k} + \dS^{\rm est}_{k-1} + \cdots + \dS^{\rm est}_{1} + \underbrace{\frac{1}{T} \sum_{\hat{\gamma}_{k}}\Prob{\hat{\gamma}_{k}} \ln \frac{\Prob{\hat{\gamma}_{k}^{(1)}}}{\Prob{\hat{\gamma}_{k}^{(k+1)}}}}_{0}\\
		&= \sum_{l=1}^k \dS^{\rm est}_{l}
	\end{align*}
\end{proof}
By choosing $n=1$ in Lemma~\ref{lem:uncond-vs-cond}, we get the following corollary, which is often implicitly used for Markov processes: In that case, $\dSM$ is some times defined with splitting probabilities and other times with length-2 path probabilities, both of which are equivalent:\vspace{0.2cm}
\begin{corollary}
	For processes of any order, we have 
	$$\dS^{\rm est}_{1} = \dS^{\rm est,uncond}_{1}.$$
\end{corollary}
Importantly, this identity is far from true for orders larger than 1. This difference can be directly quantified:\vspace{0.2cm}
\begin{corollary} \label{cor:s-cond-vs-uncond-recurrence}
	For processes of any order and for every $k\geq2$ we have
	$$ \dS^{\rm est}_{k} = \dS^{\rm est,uncond}_{k}-\dS^{\rm est,uncond}_{k-1}.$$
\end{corollary}

Further, it is straightforward to see that $\dS^{\rm est,uncond}_k$ is non-negative by the following presentation:
\begin{align}
	\dS^{\rm est, uncond}_{k}= \frac{1}{2T}  \sum_{\hat{\gamma}_{k}}\left(\Prob{\hat{\gamma}_{k}} - \Prob{\theta \hat{\gamma}_{k}}\right)
	\ln\frac{ \Prob{\hat{\gamma}_{k}}}{\Prob{\theta \hat{\gamma}_{k}}} \geq 0,
\end{align}  
since every summand is non-negative. Non-negativity of $\dS^{\rm est}$ requires a bit more work:\vspace{0.2cm}
\begin{corollary} 
	For processes of any order and for every $k\geq1$ we have
	$$ \dS^{\rm est}_{k} \geq 0$$.
\end{corollary}
\begin{proof}
	The proof for $k=1$ follows trivially from $\dS^{\rm est, uncond}_{1} \geq 0$. 
	
	For $k\geq 2$, we have
	\begin{align*}
		\dS^{\rm est}_{k} 
		&= \dS^{\rm est,uncond}_{k}-\dS^{\rm est,uncond}_{k-1} \tag{by Corollary~\ref{cor:s-cond-vs-uncond-recurrence}}\\
		&= \frac{1}{T} \sum_{\hat{\gamma}_{k}^{(1:k)}} \sum_{\hat{\gamma}_{k}^{(k+1)}}\Prob{\hat{\gamma}_{k}}
		\ln\frac{ \Prob{\hat{\gamma}_{k}}}{\Prob{\theta \hat{\gamma}_{k}}} -\dS^{\rm est,uncond}_{k-1} \tag{by Definition of $\dS^{\rm est,uncond}_{k}$}\\
		&\geq \frac{1}{T} \sum_{\hat{\gamma}_{k}^{(1:k)}}\Prob{\hat{\gamma}_{k}^{(1:k)}}
		\ln\frac{ \Prob{\hat{\gamma}_{k}^{(1:k)}}}{\Prob{\theta \hat{\gamma}_{k}^{(1:k)}}} -\dS^{\rm est,uncond}_{k-1} \tag{log-sum inequality}\\ 
		&= \dS^{\rm est,uncond}_{k-1} - \dS^{\rm est,uncond}_{k-1} = 0,
	\end{align*}
	implying the corollary.
\end{proof}
Together with Corollary~\ref{cor:s-cond-vs-uncond-recurrence}, this allows for a stronger statement about the unconditional estimator: $\dS^{\rm est,uncond}_{k}$~is not only non-negative, but also non-decreasing in $k$:\vspace{0.2cm}
\begin{corollary}
	For processes of any order and for every $k\geq1$ we have
	$$ \dS^{\rm est,uncond}_{k} \geq \dS^{\rm est,uncond}_{k-1}.$$
\end{corollary}

We were careful in stating in the corollaries that their scope are processes of \emph{any} order. In particular, the unconditional estimator $ \dS^{\rm est,uncond}_{k}$ is \emph{always} non-decreasing in $k$. $\dS^{\rm est}_{k}$, however, has been shown to stay constant once the underlying process's order, say $n$, has been reached (Lemma~\ref{lem:estimator-consistent-in-order}). The change of $\dS^{\rm est}_{k}$ for $k <n$ , however, remains unclear so far. This is equivalent to the ``change-of-change'' i.e.\ to a discrete 2nd order derivative version of $ \dS^{\rm est,uncond}_{k}$ (by Lemma~\ref{cor:s-cond-vs-uncond-recurrence}). Interestingly, as we show in the Appendix to the main Letter, no general statement is possible: $\dS^{\rm est}_{k}$  can increase as well as decrease before converging at~$k \geq n$.

\section{Apparent Power Law Scaling}
Recent works \cite{SUPyuInversePowerLaw2021,SUPcocconiScalingEntropyProduction2022,SUPyuStatespaceRenormalizationGroup2022,SUPyuDissipationLimitedResolutions2024} investigated how the energy dissipation rate in self-similar non-equilibrium reaction networks
depends on the coarse-graining scale $\lambda\equiv n/n_s$
(i.e.\ the total number of microstates $n$ relative to the number of mesoscopic states $n_s$ in a lumped system) in absence of a timescale separation and observed a power-law dependence of the Markovian estimate
$\dSM[\hGamma^s_t]\propto \lambda^{-\alpha}$
\cite{SUPyuInversePowerLaw2021}. These findings were rationalized in terms of an ``inverse cascade'' of energy dissipation on different scales. In the specific context of active flows generated by the microtubule-kinesin motor
system \cite{SUPDogic}, this amounts to the claim that the
most energy is spent to generate and maintain the flow at
smaller length scales and only a tiny amount at large
length scales \cite{SUPyuInversePowerLaw2021}. While the underlying calculations in Refs.~\cite{SUPyuInversePowerLaw2021,SUPcocconiScalingEntropyProduction2022,
	SUPyuStatespaceRenormalizationGroup2022,SUPyuDissipationLimitedResolutions2024} 
are technically sound within the inconsistent Markovian approximation, their interpretation \cite{SUPyuInversePowerLaw2021,SUPyuStatespaceRenormalizationGroup2022} is physically inconsistent, as we show in the main Letter by means of counterexamples in Figs.~3 and 5 that there are no implications between the power-law dependence of the Markovian estimate
$\dSM[\hGamma^s_t]\propto \lambda^{-\alpha}$ and the energy dissipated on different length scales.\\

\section{Splitting probabilities in lumped cycle graph \texorpdfstring{$C_n$}{Cn}} \label{sec:splitting-Cn}  

In this section, we determine the splitting probabilities in the lumped cycle graph via the Laplace transform. Consider the graph depicted in Fig.~3a in the Letter, that is, a cycle graph $C_n$ of $n$ states and with $n/\lambda \in \N$ many lumps of size $\lambda \in \N$. We focus, without loss of generality  (w.l.o.g.), on the first lump. For state $i \in [\lambda]$, we have the following system of $\lambda$ coupled differential equations
\begin{align}
	\frac{dp_i}{dt} = -(\omega + \kappa)p_i + \1_{i>1} \omega p_{i-1} + \1_{i<n} 
	\kappa p_{i+1}, \label{eq:cycle-lumping-inner-odes}
\end{align}
with absorbing boundaries 
\begin{align}
	\frac{dp_{\lambda+1}}{dt} = \omega p_\lambda, \qquad 	\frac{dp_{n}}{dt} = \kappa p_1. \label{eq:cycle-lumping-boundary-odes}
\end{align}
In Laplace space ($\tilde{p}(s) \equiv \int_{0}^{\infty} p(t) \exp(-st) dt$), we obtain
\begin{align}
	\tridiagmatrix{s+\omega+\kappa}{-\kappa}{-\omega} \begin{pmatrix}
		\tilde{p}_1(s) \\
		\tilde{p_2}(s) \\
		\vdots \\
		\vdots \\
		\tilde{p}_\lambda(s)
	\end{pmatrix}
	= \begin{pmatrix}
		p_{1}(0)\\
		p_{2}(0)\\
		\vdots \\
		\vdots \\
		p_{\lambda}(0)
	\end{pmatrix}.  \label{eq:cycle-lumping-LSE}
\end{align} 
Denoting the tridiagonal matrix in Eq.~\eqref{eq:cycle-lumping-LSE} by $M(s)$ and by $p_{j \mid i}(t)$ the probability of having reached state $j$ by time $t$ given that we started in state $i$,  we have for the ratio of splitting probabilities	
\begin{align}
	\frac{\Phi_{+ \mid +}}{\Phi_{- \mid -}} 
	&= \frac{\lim_{t \rightarrow \infty} p_{l+1 \mid 1}(t) }{\lim_{t \rightarrow \infty} p_{n \mid \lambda}(t) } \\
	&= \frac{\lim_{s \rightarrow 0} s \cdot \tilde{p}_{\lambda+1 \mid 1}(s) }{\lim_{s \rightarrow 0} s \cdot \tilde{p}_{n \mid \lambda}(s)} \tag{finite value theorem}\\
	&= \frac{\lim_{s \rightarrow 0} \omega \cdot  \tilde{p}_{\lambda \mid 1}(s) }{\lim_{s \rightarrow 0} \kappa \cdot  \tilde{p}_{1 \mid \lambda}(s)}  \tag{by Eq.~\eqref{eq:cycle-lumping-boundary-odes}}\\
	&= \frac{\omega\cdot \bra{\lambda}  \lim_{s \rightarrow 0}  M^{-1}(s)\ket{1}}{\kappa \cdot \bra{1} \lim_{s \rightarrow 0}  M^{-1}(s)\ket{\lambda}} \tag{by Eq.~\eqref{eq:cycle-lumping-LSE}} \\
	&= \left( \frac{\omega}{\kappa} \right) ^\lambda,
\end{align}
where the last step follows since we have for the tridiagonal Toeplitz matrix $M(0)$ that $\displaystyle{\frac{\left[M(0)^{-1} \right]_{\lambda,1}}{\left[M(0)^{-1} \right]_{1,\lambda}} = \left( \frac{\omega}{\kappa}\right)^{\lambda-1}}$.

\subsection{Note on the waiting time distribution} \label{sec:splitting-Cn:wtd}

Let us denote the waiting time distribution from the current lump in direction $u$ given that the preceding transition was in direction $v$ ($u,v \in \{+,-\}$) with $\Psi_{u \mid v}(t)$. By definition, we have the identity
\begin{align}
	\lim_{t \rightarrow \infty} \Psi_{u \mid v}(t) = \Phi_{u \mid v}(t).
\end{align}
From Eqs.~(\ref{eq:cycle-lumping-inner-odes}) together with (\ref{eq:cycle-lumping-boundary-odes}) we can determine the Laplace-transformed waiting-time distribution as
\begin{align}
	\tilde{\Psi}_{+ \mid +}(s) 
	&= \tilde{p}_{l+1 |1}(s)  \\
	&= \frac{\omega}{s} \bra{\lambda} M^{-1}(s) \ket{1}.
\end{align}
For instance, with $\lambda=3$, we obtain 
\begin{align}
	\tilde{\Psi}_{+ \mid +}(s)  &= \frac{\omega^3}{\poly(s)},
\end{align}
with $\poly(s)= s (s+\kappa + \omega) \left((s+ \kappa)^2+2 s \omega +\omega^2\right)$, which has roots 
\begin{align} \label{eq:cycle-lumping:s-roots-start}
	s_1 &= 0,\; s_2= -(\kappa+\omega), \\ s_3 &= -(\kappa+\omega)-\sqrt{2 \kappa\omega}, \; s_4 = -(\kappa+\omega)+\sqrt{2\kappa\omega}. \label{eq:cycle-lumping:s-roots-end}
\end{align}
Using Cauchy's residue theorem, we can invert $\tilde{\Psi}_{+ \mid +}(s)$ back to time domain
\begin{align}
	\Psi_{+ \mid +}(t) 
	&= \omega^3 \sum_{i=1}^{4} \frac{1}{\frac{d}{ds} \poly(s)|_{s=s_i}} \exp(s_i \cdot t). \label{eq:cycle-lumping:Psi-t}
\end{align}
It only remains to plug the roots $s_i$ in Eqs.~\eqref{eq:cycle-lumping:s-roots-start}-\eqref{eq:cycle-lumping:s-roots-end} into Eq.~\eqref{eq:cycle-lumping:Psi-t}. The result is a tedious linear combination of four exponentials (of which one is constant). The result is $\Psi_{+ \mid +}(t) $ for $\omega = 3, \; \kappa = 2$ is shown in Fig.~\ref{fig::cycle-lumping:psi-plus-plus-wtd} (continuous blue line) along with results from simulations (orange dots). The waiting-time distribution is clearly non-exponential, directly contradicting the assumption of Markovian dynamics on the lumped state space that is made in \cite{SUPyuInversePowerLaw2021} implicitly in using the Markovian dissipation estimate.

\begin{figure}[ht!]
	\centering
	\includegraphics{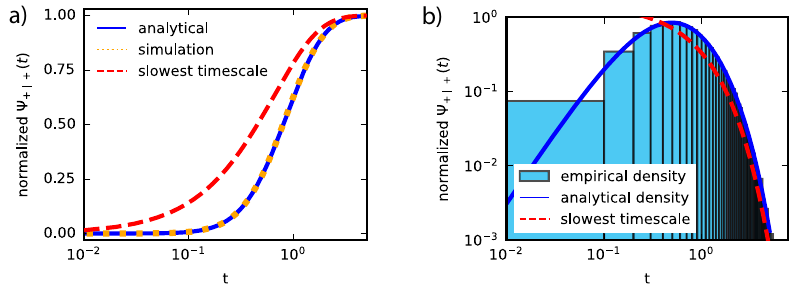}
	\caption{Highly non-exponential waiting time density $\Psi_{+ \mid +}(t)$ for the lumped cycle graph with lump size $\lambda=3$ and parameters  $\omega=3,\; \kappa = 2$ normalized by the splitting probability: (a) Distribution functions where the analytical result (blue line) matches the simulation (dotted orange line). While the slowest timescale is exponential (red dashed line), the waiting time density is not. (b) Probability density functions in log-log scale. The exponential, slowest time scale can be seen as a cut-off (dashed red) of the non-exponential waiting time density.}
	\label{fig::cycle-lumping:psi-plus-plus-wtd}
\end{figure}

For the calculation of the waiting time distribution in the reverse direction,  $\Psi_{- \mid -}(t)$, we proceed analogously. In particular, we note that the polynomial is the same, hence its roots and derivative are equal, and the only difference is the prefactor
\begin{align}
	\Psi_{- \mid -}(t) 
	&= \kappa^3 \sum_{i=1}^{4} \displaystyle{\frac{1}{\frac{d}{ds} \poly(s)|_{s=s_i}}} \exp(s_i \cdot t).
\end{align}
We hence have that the conditional waiting time distribution is equal in both directions
\begin{align}
	\frac{\Psi_{+ \mid +}(t) }{\Phi_{+ \mid +}} = 	\frac{\Psi_{- \mid -}(t) }{\Phi_{- \mid -}}, \quad \forall t \in \R^+.
\end{align}
Thus, in particular, their Kullback-Leibler divergence vanishes $\displaystyle{D_{\rm KL}\left[\left.	\frac{\Psi_{+ \mid +}(t) }{\Phi_{+ \mid +}} \right\| 	\frac{\Psi_{- \mid -}(t) }{\Phi_{- \mid -}} \right] = 0}$, so that waiting times do \emph{not} contribute to the entropy production rate, which is in turn fully captured by the splitting probabilities \cite{SUPmartinezInferringBrokenDetailed2019}.

\section{Tree graph}
We now consider the lumped tree of (edge) depth $d$ with lumps of depth $l$. Let the underlying microscopic graph $G_T$ be constructed as follows.

\subsection{Construction of tree graph}
Consider a tree graph of depth $d$, i.e., vertices are on levels ${U_0, \ldots, U_d}$, where level $i$ contains ${|U_i| = 2^i}$ vertices. For $i<d$, each vertex ${v\in U_i}$ has two offsprings ${x,y \in U_{i+1}}$ with rates ${L_{xv}=L_{yv}=\omega/2^{i+1}}$ in $+$ direction and ${L_{vx}= L_{vy}= \kappa/2^{i+1}}$ in $-$ direction.~Vertices in the bottom level $v\in U_d$ have edges of rate $\omega/2^d$ towards the root, and rate $\kappa/2^d$ from the root to each $v$.
We refer to Fig.~3b in the Letter. The steady state is the uniform distribution as can be readily seen from the generator $L$.

\subsection{Microscopic dynamics}
The entropy production rate of the microscopic system can be computed with a cycle decomposition. Let  $T$ be the spanning tree obtained from $G_T$ with all edges between levels $U_d$ and $U_0$ removed. Let $\mathcal C$ be the cycle basis induced by $T$ \cite{SUPschnakenbergNetworkTheoryMicroscopic1976}. For any $c \in \mathcal C$, the cycle's current is determined by the flux of the chord -- i.e. the edge between $U_d$ and $U_0$, yielding
\begin{align}
	J_c = \frac{\omega-\kappa}{n 2^d},
\end{align}
where the factor $n^{-1}$ stems from the steady state and is determined through depth $d$.

The affinity of $c$ is 
\begin{align}
	A_c = (d+1) \ln\left(\frac{\omega}{\kappa}\right).
\end{align}
Since the dimension is $|\mathcal C|=2^d$, we get an entropy production rate
\begin{align}
	\dot S = \sum_{c \in \mathcal C}\dot S(c) =  \sum_{c \in \mathcal C} J_c A_c = \frac{d+1}{2^{d+1}-1}(\omega-\kappa)\ln\left(\frac{\omega}{\kappa}\right), \label{eq:tree:Sdot-micro}
\end{align}
where in the last step we used that the number of vertices is $n(d) = 2^{d+1}-1$.

\subsection{Mistaken dynamics}
We now analyze what happens if we mistakenly assume Markov dynamics of the lumped state space with lump depth $l$.
Since the steady-state distribution is uniform over all vertices, a lump of depth $l$ has a steady-state probability 
\begin{align}
	\Prob{\triangle}= \left( 2^{l+1}-1\right) \frac{1}{n}.
\end{align}
Suppose the triangle is rooted in level $U_i$. According to Eq.~(4) in~\cite{SUPyuInversePowerLaw2021} the upwards ``effective rate'' from the triangle is
\begin{align}
	k_{\triangle_i \uparrow} = \frac{\kappa}{2^i \cdot  \left( 2^{l+1}-1\right)},
\end{align}
while the rate to a neighboring lump downward reads
\begin{align}
	k_{\triangle_i \downarrow} = \frac{\omega}{2^{i+l+1} \cdot  \left( 2^{l+1}-1\right)}.
\end{align}	
For an edge between two tree lumps rooted at levels $U_{i-l-1}$ and $U_i$, we have
\begin{align}
	&\dot S_{\text{single edge }\triangle_{i-l-1} \leftrightarrow \triangle_{i}} \\
	=&  \Prob{\triangle} \left( k_{\triangle_{i-l-1} \downarrow} - k_{\triangle_{i} \uparrow} \right)\ln \left(\frac{k_{\triangle_{i-l-1} \downarrow}}{k_{\triangle_{i} \uparrow}}\right)\\
	=& \frac{1}{n 2^i} \left(\omega - \kappa \right) \ln \left(\frac{\omega}{\kappa}\right).
\end{align}
There are $2^i$ such edges, which contribute 
\begin{align}
	\dot S_{\text{all edges on level }\triangle_{i-l-1} \leftrightarrow \triangle_{i}} =  \frac{1}{n} \left(\omega - \kappa \right) \ln \left(\frac{\omega}{\kappa}\right).
\end{align}
There are $\displaystyle{\frac{d+1}{l+1}}$ such transitions between levels, so that we get with $n \equiv n(d) = 2^{d+1}-1$ an entropy production rate
\begin{align}
	\dot S_M &= \frac{d+1}{2^{d+1}-1} \frac{1}{l+1} (w - k)   \ln\left(\frac{\omega}{\kappa}\right)\\
	&=  \frac{d+1}{2^{d+1}-1} \frac{\ln(2)}{\ln(\lambda+1)} (w - k)   \ln\left(\frac{\omega}{\kappa}\right) \propto \frac{1}{\ln(\lambda+1)},
\end{align} 
where  in the last step we used that $\lambda \equiv n_0/n_s = 2^{l+1}-1$.
This result is illustrated for the case $d=23$ (and thus $n=16777215$) in Fig.~3d in the letter.

\subsection{Mesosystem}
We now derive the entropy production rate of the lumped mesosystem when memory is accounted for. Via Laplace transforms, we obtain the splitting probabilities. We defer the derivation to the more general setting in Section~\ref{sec:diamond-tree-details}. For the special setting of the tree graph, we obtain using Eq.~\ref{eq:splittings-tree-lump-in-tree-diamond} for the ratio of splitting probabilities
\begin{align}
	\frac{\Phi_{+ \mid +}}{\Phi_{- \mid -}} = \left(\frac{ \omega}{\kappa} \right)^{l+1}.
\end{align} 
There are $\frac{d+1}{l+1}$ lump levels. The net flux between each level is $J=\frac{\omega-\kappa}{n}$, yielding a $2^{\rm nd}$ order estimate of the entropy production rate
\begin{align}
	\dot S_2^{\rm est} = J \ln\left(\frac{\Phi_{+ \mid +}}{\Phi_{- \mid -}}\right) = (d+1) \frac{\omega -\kappa}{2^{d+1}-1} \ln\left(\frac{ \omega}{\kappa} \right),
\end{align}
which is exact by Eq.~\eqref{eq:tree:Sdot-micro}.

\section{Sierpinski-type graph}
In this section, we first calculate the dissipation rate on the fractal Sierpinski-type graph of depth $d=3$ on each length scale and then show a comparison of the estimator for the entropy production with the analytic solution up to depth $d=6$ ($n=3072$ vertices).

\subsection{Decomposition in cycle basis}
We decompose the Sierpinski-type~graph of depth $d=3$ in a cycle basis. To do so, we consider the spanning tree depicted as in Fig.~\ref{fig:sierpinski-spanningtree}. Recall from the Letter that we have polygons driven by rate $\omega$ at different depths: At the finest scale (depth $1$ subgraphs), the triangles are driven with rate $\omega$. At the next scale (depth $2$) the hexagons are driven by $\omega$. Finally, the outermost regular polygon, the nonagon (depth $3$), is also driven by $\omega$. See Fig.~5a in the Letter for an illustration of the driving. All other edges (and the backward edges of the driven edges) have a transition rate $1$.

\begin{figure}[hbt]
	\includegraphics[width=0.5\linewidth]{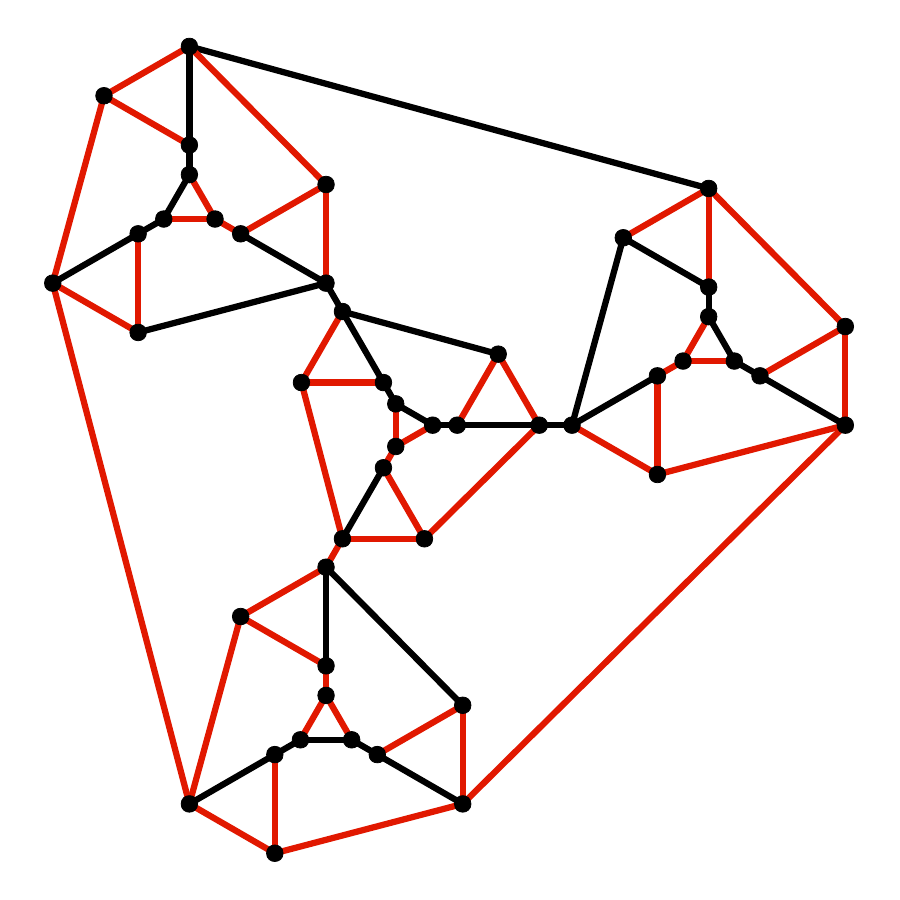}
	\caption{Spanning tree for Sierpinski-type graph of depth $d=3$ inducing the considered cycle decomposition.}
	\label{fig:sierpinski-spanningtree}
\end{figure}

The elements of the induced cycle basis with non-zero current are depicted in Fig.~\ref{fig:sierpinski-treedecomp}.  These basis elements coincide with the driven cycles at depth $1$ (blue), depth $2$ (green), and depth $3$ (orange). While the inner subgraph $D$ at depth 2 recomposes perfectly into its driven cycles, each outer subgraph at depth 2 has in addition two corrective cycles, which will be hidden at coarse~grainings of scale $\lambda \geq 12$ (depth $\geq 2$).

\begin{figure}[hbt]
	\includegraphics[width=0.5\linewidth]{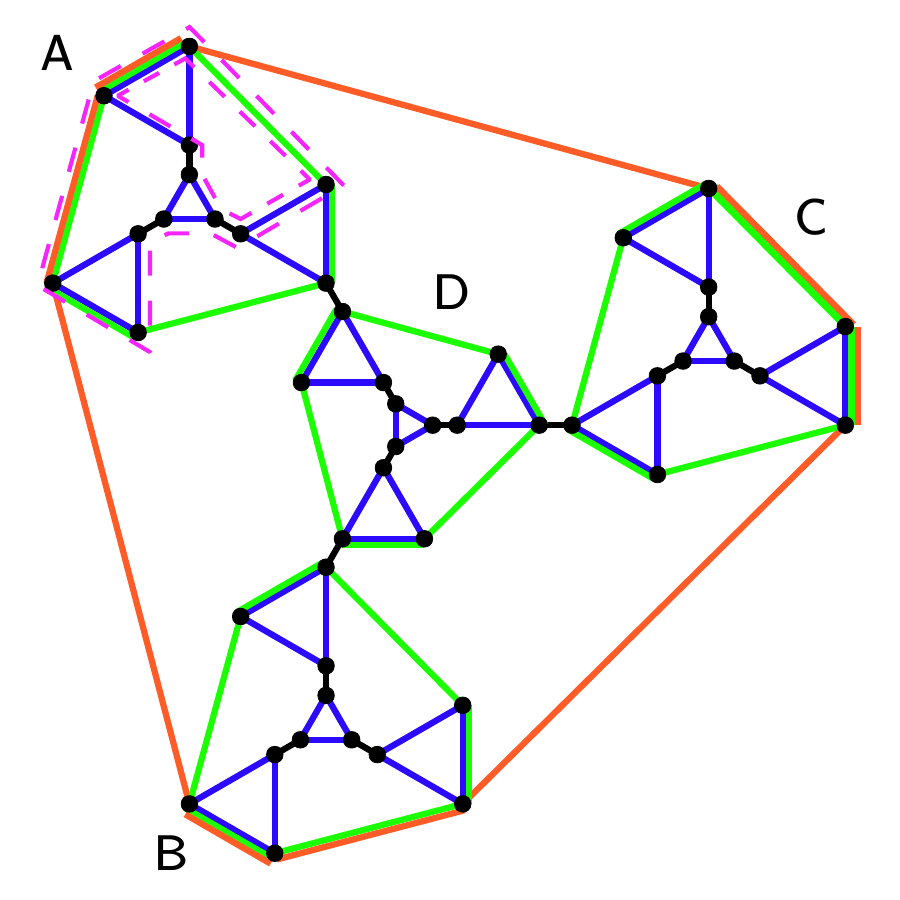}
	\caption{Basis cycles with non-zero current. These basis elements coincide with the driven cycles of depth $1$ (blue), depth $2$ (green), and depth $3$ (orange) as well as two corrective cycles in each of the outer subgraphs A, B, and C. These correction cycles are indicated by dashed lines in subgraphs $A$.}
	\label{fig:sierpinski-treedecomp}
\end{figure} 

A driven regular cycle at depth $l$ has an affinity of $A_l = 3l \log{\omega}$. When driven with $\omega=2$, the contribution of each scale is shown in Table.~\ref{tab:cycle-basis-contribution}.

\begin{table}[hbt]
	\caption{Contribution of cycle basis elements at different scales.}
	\begin{tabular}{l|cc}
		Basis cycles & Contribution to $\dot S$ \\
		\hline
		Depth 1& 0.59 \\
		Depth 2& 0.17\\
		Depth 3& 0.07\\
		Corrective cycles & -0.01
	\end{tabular}
	\label{tab:cycle-basis-contribution}
\end{table}

\subsection{Comparison of stochastic estimator to analytic solution}

We compare the entropy production rate obtained by the stochastic estimator for the microscopic graph to the analytic entropy production rate obtained by solving for the nullspace of the estimator, at a driving of $\omega=2$. Note that, remarkably, although the graph becomes very large (at depth $d=6$ we have $n=3072$ many vertices), the matrix is sparse enough to find the steady state analytically at that size. 
The results summarized in Table~\ref{tab:comparison-our-est-to-anlytic} confirm that the estimator is indeed very precise.
\begin{table}[hbt]
	\caption{Estimated entropy production rate $\dot S^{\rm est}_1$ vs.\ analytic solution $\dot S$ for Sierpinski-type graphs of different depth. The estimation is based on trajectories of length $10^8$ steps, after skipping $10^8$ initial "equilibration" steps.}
	\begin{tabular}{l|ccc}
		Depth of Sierpinski-type graph &$n$ & $\dot S^{\rm est}_1$ & $\dot S$ \\ \hline 
		2 & 12   & 0.8011 & 0.8010\\
		3 & 48   & 0.8255 & 0.8259\\
		4 & 192  & 0.8325 & 0.8323\\
		5 & 768  & 0.8349 & 0.8341\\
		6 & 3072 & 0.8349 & 0.8346\\
	\end{tabular}
	\label{tab:comparison-our-est-to-anlytic}
\end{table}

\section{Brusselator}
For the Brusselator we choose the parameters as in \cite{SUPyuInversePowerLaw2021}. Note that the state space can be reduced from $500 \times 500$ states to e.g.\ $450 \times 450$ states due to the location of the limit cycle in the steady state. This can be verified numerically, by plotting the average time spent in a given state in the steady state. We show this statistic in Fig.~\ref{fig:brusselator-histogram} for a trajectory of length $5*10^8$ after skipping an initial phase of $10^9$ steps ensuring convergence to the steady state, where in particular the right and top region is $0$ (black). Indeed, summing the steady-state probabilities for $x>450, y>450$, we get $2.45e-6$.

In the Sierpinski-type graph, we were able to exploit the sparsity of the generator to find a basis for the null space. However, in the case of the Brusselator, this corresponds to finding the basis of the null space of a $202500 \times 202500$ matrix. We are not aware of any computational solution providing reliable results without significant errors for such matrix sizes, therefore we resort to stochastic simulations.
\begin{figure}[hbt]
	\includegraphics[width=0.5\linewidth]{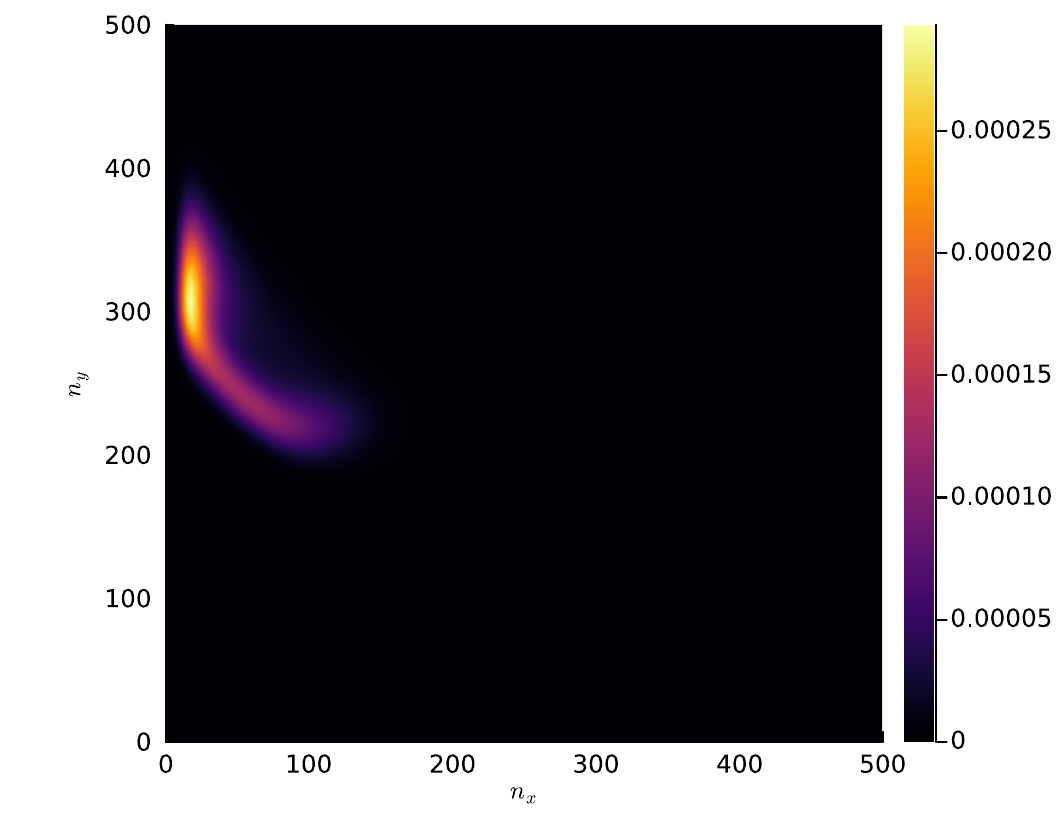}
	\caption{Average time spent in a given state in the steady state of the Brusselator. The results were inferred from a trajectory of length $5*10^8$ upon skipping $10^9$ steps to reach the steady state.}
	\label{fig:brusselator-histogram}
\end{figure}

\section{Tree diamond graph}\label{sec:diamond-tree-details}

\begin{figure}[ht!!]
	\centering
	\includegraphics[width=0.45\textwidth]{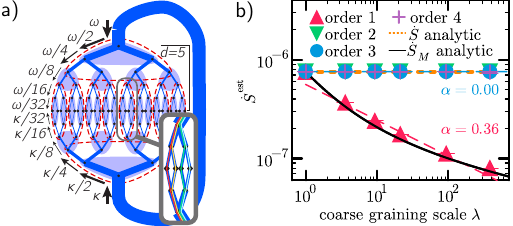}
	\caption{(a)~Tree diamond graph of depth $d=5$ with lumps of depth $l=1$ (shaded in blue) and depth $l=2$ (circled in dashed red).~Transition rates in down ($+$) and up ($-$) direction (cyclic with respect to the thick blue edge) are taken to be   $\propto \omega$ and $\propto \kappa$, respectively, chosen such that the cumulative transition rates between each level is $\omega - \kappa$.~The self-similarity becomes obvious upon noticing that each tree diamond consists of multiple smaller tree diamonds.~Inset: Four fundamental cycles pass through the highlighted subgraph of depth $d=2$.~(b)~Entropy-production estimates for different coarse-graining scales; Estimators $\dS^{\rm est}_k$ provide an exact estimation of $\dS$ for higher orders $(k\geq 2$), matching the analytical prediction (yellow line).~Ignoring memory leads to a \emph{virtual} energy dissipation cascade, which has \emph{no} implications
		for dissipation on different scales.~The virtual cascade is actually \emph{not} even a power law, but may be mistaken for one.}  
	\label{fig:tree-diamond}
\end{figure}

We now consider a non-trivially self-similar network with a single dissipative scale, the system size , but multiple (equivalent) dissipative cycles each with a stationary
current $J_c$ (see Fig.~\ref{fig:tree-diamond}a). It corresponds to a ``tree diamond''  obtained by merging two binary trees of depth $d$ such that they have the same set of leaves.
Every edge is bidirectional with a rate $\propto\omega$ in $+$ and $\propto \kappa$ in $-$ direction, respectively, such that the cumulative transition rate between each level is $\omega-\kappa$.  This way, the edge entering the root in $+$ direction has rate $\omega$, while the rates along the remaining edges are split evenly amongst the branches, (see Fig.~\ref{fig:tree-diamond}a).  The detailed mathematical construction is described in Section~\ref{sec:tree-diamond:construction}. The exact $\dS$ of the system in the cycle basis $\mathcal{C}$ with $2^d$ basis cycles reads (proof in Section~\ref{sec:tree-diamond-microsopic-entropy-production-rate}).
\begin{align}
	\dS = \sum_{c\in \mathcal{C}}=J_cA_c=\frac{2d+1}{n(d)}(\omega - \kappa)\ln
	\left(\frac{\omega}{\kappa}\right).
	\label{Diamond_tree_exact}
\end{align}
where $n(d)=2^{d+1}-1+2^{d}-1=3\cdot 2^d-2$ is the number of states in a diamond with depth $d$, and the basis-cycle current (see \cite{SUPschnakenbergNetworkTheoryMicroscopic1976}) reads $J_c=(\omega - \kappa)/n(d)2^d$.

We lump the network into lumps of depth ${l<d}$, whereby the $2^{d-l}$ distinct lumps on the ``equator'' are treated separately (see~Fig.~\ref{fig:tree-diamond}b). The Markovian estimate of such a lumped system is given by (see Section~\ref{sec:tree-diamond-mistaken-entropy-production-rate})
\begin{align}
	\!\!\dSM\! &=\! \frac{\omega - \kappa
	}{n(d)} 
	\frac{2d-l+1}{l+1}
	\ln \left(\frac{\omega}{\kappa}\right)\overset{{d\gg l\gg
			1}}{\propto}\frac{1}{l}\propto \frac{1}{\ln(\lambda/3)},
	\label{Diamond_tree_Markov}
\end{align}
where $\lambda=n/n_s$ now corresponds, by definition, to the average lump size.~As in the tree graph, there is in fact strictly \emph{no} power-law scaling of $\dSM$ in $\lambda$ despite a perfectly self-similar network structure. Moreover, as shown in Fig.~\ref{fig:tree-diamond}b, one may in practice easily identify a power-law $\sim \lambda^{-0.33}$  in agreement with
\cite{SUPyuInversePowerLaw2021}. As in the tree graph, the microscopic system has a single dissipative scale, further underscoring that an observed power law has \emph{no} implications for energy dissipation on distinct scales and may \emph{not} even be a power law. It is an artifact of inconsistent time reversal in the presence of memory.

The lumping yields a $2^{\rm nd}$-order semi-Markov process and no dissipative cycles become hidden. By accounting for memory via $\dS_{k\ge 2}^{\rm est}$ we
recover the microscopic result (see derivation in the subsequent subsection)
\begin{align}
	\!\!\!\dS_2^{\rm est}&\!=\!\frac{\omega-\kappa}{n(d)}\!\left[2\frac{d-l}{l+1}\ln\left(\frac{\Phi^{\triangle}_{+|+}}{\Phi^{\triangle}_{-|-}}\right)\!+\!\ln\left(\frac{\Phi^{\diamondsuit}_{+|+}}{\Phi^{\diamondsuit}_{-|-}}\right)
	\right]\!=\!\dS,
	\label{2_diamond_tree}
\end{align}
where $\Phi^{\triangle}_{\pm|\pm}$ is the conditional probability for a transition between two lumps (at given coarseness) in the $\pm$ direction given that the preceding step occurred in the $\pm$ direction in the top and bottom part of the network (the factor of two is due to symmetry), and $\Phi^{\diamondsuit}_{\pm|\pm}$ the corresponding conditional probability for
a transition in the $\pm$ direction in any of the equatorial lumps. By the last equality
consistent time reversal (and thus thermodynamics) removes the artefactual scaling behavior.

\subsection{Construction} \label{sec:tree-diamond:construction}
The tree diamond is constructed from two perfect binary trees of edge depth $d$, the top one rooted at the top, the bottom one rooted at the bottom, merging them to have the same set of leaves, and connecting the two roots with a bidirectional edge. Fig.~\ref{fig:tree-diamond}a in the Letter shows a tree diamond of depth $d=5$.  Hence, the tree diamond has
\begin{align}
	n\equiv n(d) =3\cdot 2^{d}-2 \label{eq:tree-diamond:n}
\end{align}
vertices. In this way, we obtain edges on levels $U_0, \ldots, U_{2d}$:  Levels $U_0, \ldots U_{d-1}$ are the levels of the upper tree, level $U_d$ are the child nodes of both trees and levels $U_{d+1}, \ldots, U_{2d}$ are the levels of the bottom tree. Hence, the very top level $U_0$ consists of only one vertex, as does the very bottom level $U_{2d}$. The rates along edges between levels of the top tree halve each step and double between levels of the bottom tree.

When walking on the tree, we go in \emph{positive} direction, if we go from a level $U_i$ to $U_{i+1 \mod 2d+1}$.
Hence, the rates of the edges between levels $U_{i}$ and $U_{i+1 \mod 2d}$ sum to $\omega$ in $+$ direction and $\kappa$ in $-$ direction. From $0 =  L p$ we see that the steady state is the uniform distribution over all vertices.

For any $l\in \N$ with $d+1 = 0 \mod l+1$,  we perform a coarse-graining of depth $l$ as follows. We start from the top root, lumping trees of depth $l$ as we go downwards towards the leaves. We do the same from the bottom root upwards. Those up- and down-facing trees in the center, which share the same set of leaves, are merged pairwise into one lump, which is again a smaller diamond tree. A lumping of depth $l=1$ and $l=2$ is shown in Fig.~\ref{fig:tree-diamond}a in the Letter in shaded blue and dashed red, respectively.

\subsection{Microscopic dynamics} \label{sec:tree-diamond-microsopic-entropy-production-rate}
We provide the following two different ways of estimating the \emph{microscopic} entropy production rate $\dot S$: One via the direct formula for Markovian systems which was used by Yu et al.\ \cite{SUPyuInversePowerLaw2021} for \emph{both} Markovian and non-Markovian systems. The second derivation is based on the graph decomposition in its cycle basis and hence is more instructive for the understanding of dissipative cycles. 

\subsubsection{\texorpdfstring{$\dot S$}{dS} via edges}
Let us consider an edge in the top tree, i.e.\ from level $U_i$ to level $U_{i+1}$ with $i\in \{0, \ldots, d-1\}$.
This edge contributes to the entropy production rate as 
\begin{align}
	\dot S_{\text{single edge }U_i \leftrightarrow U_{i+1}} = \frac{1}{n} \frac{1}{2^{i+1}} (\omega - \kappa )\ln \left(\frac{\omega}{\kappa}\right),
\end{align}
where the first factor comes from the uniform steady-state probability distribution and the second from the transition rates. 
There are $2^{i+1}$ such edges, so that 
\begin{align}
	\dot S_{\text{level }U_i \leftrightarrow U_{i+1}} = \frac{1}{n} (\omega - \kappa )\ln \left(\frac{\omega}{\kappa}\right).
\end{align}
Since there are $2d$ such levels (edges in the bottom half contribute equally by symmetry) as well as one edge connecting the bottom to the top, we have an overall entropy production rate of
\begin{align}
	\dot S = \frac{2d+1}{n(d)}(\omega - \kappa)\ln \left(\frac{\omega}{\kappa}\right),\label{eq:diamond-Sdot-micro-appendix}
\end{align}
which is stated in the Letter. 

\subsubsection{Cycle decomposition}
Let our diamond tree be the graph $G=(V,E)$. $G$ has a cycle basis $\mathcal C = \{c_1, \ldots, c_{2^d}\}$ of size $|\mathcal C| = 2^d$, with each basis cycle $c_i$ being the unique cycle of length $2d+1$ going through a unique node in level $L_d$. To see this, take as spanning tree $T$ the graph $G$ with all edges between level $L_d$ and $L_{d+1}$ removed. Adding any single edge between $L_d$ and $L_{d+1}$ will create a cycle. We call the added edge a \emph{chord} and identify uniquely each cycle $c_i$ with its chord $h_i \in E$. 
Four of these cycles are highlighted at the central levels in Fig.~\ref{fig:tree-diamond}a (inset) of the Letter.

We focus on a cycle $c_i$. This cycle's flux is determined by the flux of its chord $h_i$:
\begin{align}
	J_{c_i} = \frac{\omega - \kappa}{n2^d},
\end{align}
where the factor $n^{-1}$ stems from the steady state and $2^{-d}$ from the transition rates.
The affinity of $c_i$ is
\begin{align}
	A_{c_i}=(2d+1) \ln\left(\frac{\omega}{\kappa}\right),
\end{align}
since the cycle has length $|c_i| = 2d+1$.
Thus, the cycle's entropy production rate is \cite{SUPschnakenbergNetworkTheoryMicroscopic1976}
\begin{align}
	\dot S_{c_i} = J_{c_i} \cdot A_{c_i} = \frac{2d+1}{n 2^d} (\omega-\kappa) \ln\left(\frac{\omega}{\kappa}\right).
\end{align}
Since there are $2^d$ basis cycles, we have an overall entropy production rate of 
\begin{align}
	\dot S  = \sum_{c_i}  J_{c_i} \cdot A_{c_i} = \frac{2d+1}{n(d)} (\omega-\kappa) \ln\left(\frac{\omega}{\kappa}\right),
\end{align}
where we indicated with $n(d)$ that $n$ is determined by $d$, cf.\ Eq.~\eqref{eq:tree-diamond:n}.

\subsection{Mesosystem}
We now consider a coarse-graining of the tree diamond graph yielding a $2^{\rm nd}$-order semi-Markov process. Hence, we need to condition the previous, current, and next mesoscopic states to have fully defined waiting times. We hence introduce 
$\Psi_{x \mid u , y}(t)$ for $x,y \in \{+,-\},\; u \in V$ as the distribution of transitioning from $u$ in $x$ direction by time $t$ given that we transitioned to $u$ in $y$ direction at time $t=0$. The splitting probability is the limit
$\Phi_{x \mid u , y}\equiv \lim_{t \rightarrow \infty}\Psi_{x \mid u , y}(t)$.

We now determine $\Phi_{+ \mid u , +},\Phi_{- \mid u , -}$ for $u$ being the entry/exit point of both (a) tree and (b) tree-diamond lumps.

\subsubsection{Splitting probability \texorpdfstring{$\Phi$}{Phi} for tree lump}
Consider a tree of (edge) depth $l$, which is w.l.o.g.\ rooted at its top vertex. We label the vertex levels \emph{within that tree} by $W_i,\; i\in \{0,\ldots, l\}$ so that level $W_i$ has $|W_i| = 2^i$ vertices.
By symmetry, it suffices to consider $P_{W_i}(t)$, the probability of being in \emph{any} of the vertices of level $W_i$. 
Let $(\lambda \omega, \lambda \kappa)$ for $\lambda \in \{2^{-k} \mid k \in \N_0 \}$ be the weight of the edge entering the tree's root (i.e.\ $W_0$) from above. The choice of $k$ depends on the level $U_j$ of the diamond tree in which the considered tree is rooted.

We can describe the evolution by the following system of ordinary differential equations. For $i \in \{1, \ldots, l-1\}$, we have 
\begin{align}
	\frac{dP_{W_i}(t)}{dt} =\lambda \left(  \omega P_{W_{i-1}}(t) + \kappa P_{W_{i+1}}(t)- (\kappa + \omega) P_{W_i}(t) \right),
	\label{eq:diamondtree:tree:odeLi}
\end{align}
whereas for the boundaries $i\in\{0,l\}$, we have 
\begin{align}
	\frac{dP_{W_0}(t)}{dt} &=\lambda \left( \kappa P_{W_{1}}(t)- (\kappa + \omega) P_{W_0}(t) \right)	\label{eq:diamondtree:tree:odeL0}\\
	\frac{dP_{W_l}(t)}{dt} &=\lambda \left(\omega P_{W_{l-1}}(t)- (\kappa + \omega) P_{W_l}(t) \right). 	\label{eq:diamondtree:tree:odeLd}
\end{align}
Since we consider first-passage times, we place an absorbing state $S$ at the lower-labeled end, and an absorbing state $D$ at the higher-labeled end with
\begin{align}
	\frac{dP_{S}(t)}{dt} &= \lambda \kappa P_{W_{0}}(t)\\
	\frac{dP_{D}(t)}{dt} &= \lambda \omega P_{W_{l}}(t). \label{eq:diamondtree:tree:odeLDRAIN}
\end{align}
Since these two states do \emph{not} influence the probability distribution of the non-absorbing states, we first solve the system of non-absorbing states and then for the absorbing states.
The initial distribution is either $\1_{W_0}$ (in the $+ \mid +$ setting) or  $\1_{W_l}$ (in the $- \mid -$ setting).

To solve the system of $l+1$ differential equations, we Laplace transform ($\tilde{P}_{W_i}(s)\equiv \int_{0}^{\infty} P_{W_i}(t) \exp(-st)dt$). Solving for the Laplace-transformed distributions reduces to solving the following tridiagonal linear system of equations
\begin{align}
	\tridiagmatrix{s+\lambda \kappa +\lambda \omega }{-\lambda \kappa }{-\lambda \omega } 
	\begin{pmatrix}
		\tilde{P}_{W_0}(s)\\
		\\
		\vdots \\
		\\
		\tilde{P}_{W_l}(s)
	\end{pmatrix} = 
	\begin{pmatrix}
		a\\
		0\\
		\vdots \\
		0	\\
		1-a
	\end{pmatrix},
	\label{eq:diamondtree:tree:LSE}
\end{align}
where $a \equiv 1$ iff we start at the root (in the $+ \mid +$ setting) and $a \equiv 0$ else (in the $- \mid -$ setting).

For any $l$, we can solve the above equation to obtain $\tilde{P}_{W_i}(s)$ explicitly. It remains to perform an inverse Laplace transform, for which we use Cauchy's residue method. We show the result here for $l=2$ in the $+ \mid +$ setting, the "program" being analogous for every other value of $l \in \N$ and directional conditions.
Solving Eq.~\eqref{eq:diamondtree:tree:LSE} yields 
\begin{align}
	\tilde{P}_{W_2}(s) = \frac{\lambda^2 \omega^2}{(s+ \lambda \kappa + \lambda \omega)\left[ (s+ \lambda \kappa + \lambda \omega)^2 -2 \kappa \lambda^2 \omega \right]}.
\end{align}
The roots of the denominator are 
\begin{align}
	s_1 \equiv -\lambda (\kappa+ \omega) \\
	s_{2/3} \equiv  -\lambda (\kappa+ \omega) \pm  \lambda \sqrt{2\kappa \omega},
\end{align}
yielding residue values of $-\frac{\omega}{2\kappa}, \frac{\omega}{4 \kappa}, \frac{\omega}{4\kappa}$ (the first derivative of the denominator is non-zero for all of $s_1, s_2, s_3$). Thus

\begin{align}
	P_{W_2}(t) = \frac{\omega}{2\kappa} \left\{-\exp\left[-\lambda (\kappa+ \omega)  t \right]+ \frac{1}{2}\exp\left[-\lambda (\kappa+ \omega) t-  \lambda \sqrt{2\kappa \omega} t \right] + \frac{1}{2}\exp\left[-\lambda (\kappa+ \omega) t+  \lambda \sqrt{2\kappa \omega} t \right]  \right\},
\end{align}
and hence by Eq.~\eqref{eq:diamondtree:tree:odeLDRAIN}
\begin{align}
	\frac{dP_D(t)}{dt} =  \frac{\lambda \omega^2}{2\kappa} \left\{-\exp\left[-\lambda (\kappa+ \omega)  t \right]+ \frac{1}{2}\exp\left[-\lambda (\kappa+ \omega) t-  \lambda \sqrt{2\kappa \omega} t \right] + \frac{1}{2}\exp\left[-\lambda (\kappa+ \omega) t+  \lambda \sqrt{2\kappa \omega} t \right]  \right\}.
\end{align}
This is the density of exit through $D$ at time $t$, given entry into the tree lump at its root $W_0$ at time $0$. Thus,
\begin{align}
	\psi_{+\mid\triangle,+}(t)
	&\equiv \frac{d}{dt}\Psi_{+\mid\triangle,+}(t)
	=\frac{dP_D(t)}{dt} \notag \\ 
	&= \frac{\lambda \omega^2}{2\kappa} \left\{\!-\exp\left[-\lambda (\kappa\!+\! \omega)  t \right]+ \frac{1}{2}\exp\left[-\lambda (\kappa+ \omega) t-  \lambda \sqrt{2\kappa \omega} t \right] + \frac{1}{2}\exp\left[-\lambda (\kappa+ \omega) t+  \lambda \sqrt{2\kappa \omega} t \right] \! \right\}.
\end{align}
We may proceed analogously with the initial condition $a=0$ to obtain a solution for $P_{W_0}(t)$ and
\begin{align}
	\psi_{-\mid\triangle,-}(t)
	&\equiv\frac{d}{dt}\Psi_{-\mid\triangle,-}(t)
	=\frac{dP_S(t)}{dt} \notag \\
	&= \frac{\lambda \kappa^2}{2\omega} \left\{-\exp\left[-\lambda (\kappa+ \omega)  t \right]+ \frac{1}{2}\exp\left[-\lambda (\kappa+ \omega) t-  \lambda \sqrt{2\kappa \omega} t \right] + \frac{1}{2}\exp\left[-\lambda (\kappa+ \omega) t+  \lambda \sqrt{2\kappa \omega} t \right]  \right\}.
\end{align}
since the roots $s_i$ have the same values as in the positive direction case~\footnote{Even stronger, inverting the matrix from Eq.~\eqref{eq:diamondtree:tree:LSE} yields the same denominator for each entry, hence we get the same roots, and hence the same powers of the exponentials in the waiting time distribution.}. These exit-time densities decay to zero, while their cumulative integrals converge to the splitting probabilities.

Since $P_D(0)=P_S(0)=0$, the cumulative exit probabilities are
\begin{align*}
	\Psi_{+\mid\triangle,+}(t)
	=\int_0^t\psi_{+\mid\triangle,+}(u)\,du,
	\qquad
	\Psi_{-\mid\triangle,-}(t)
	=\int_0^t\psi_{-\mid\triangle,-}(u)\,du.
\end{align*}
For general $l$, let $M_l(s)$ denote the tridiagonal matrix in
Eq.~\eqref{eq:diamondtree:tree:LSE}. Integrating the exit-time
densities gives
\begin{align*}
	\Phi_{+\mid\triangle,+}
	=\lambda\omega[M_l(0)^{-1}]_{l0},
	\qquad
	\Phi_{-\mid\triangle,-}
	=\lambda\kappa[M_l(0)^{-1}]_{0l}.
\end{align*}

The cofactor formula gives
\begin{align*}
	[M_l(0)^{-1}]_{l0}
	=\frac{(\lambda\omega)^l}{\det M_l(0)},
	\qquad
	[M_l(0)^{-1}]_{0l}
	=\frac{(\lambda\kappa)^l}{\det M_l(0)}.
\end{align*}

In the exemplary case above, we chose $l=2$. For general $l$, the ratio of splitting probabilities yields
\begin{align}
	\frac{\Phi_{+ \mid \triangle, +}}{\Phi_{- \mid \triangle, -}} = \frac{\lim_{t \rightarrow \infty} 	\Psi_{+ \mid \triangle , +}(t) }{\lim_{t \rightarrow \infty} 	\Psi_{- \mid \triangle , -}(t) } = \left(\frac{ \omega}{\kappa} \right)^{l+1}. \label{eq:splittings-tree-lump-in-tree-diamond}
\end{align}
For a triangle rooted at its highest-labelled microstate ($\triangledown$), we obtain by symmetry
\begin{align}
	\frac{\Phi_{+ \mid \triangledown, +}}{\Phi_{- \mid \triangledown, -}} = \frac{\lim_{t \rightarrow \infty} 	\Psi_{+ \mid \triangledown , +}(t) }{\lim_{t \rightarrow \infty} 	\Psi_{- \mid \triangledown , -}(t) } = \left(\frac{\kappa}{\omega} \right)^{-(l+1)} = \frac{\Phi_{+ \mid \triangle, +}}{\Phi_{- \mid \triangle, -}}.
\end{align}
We corroborate the analytical results with computer experiments in Fig.~\ref{fig:tree-affinities}, and find that they are in excellent agreement.
\begin{figure}[ht]
	\centering
	\includegraphics{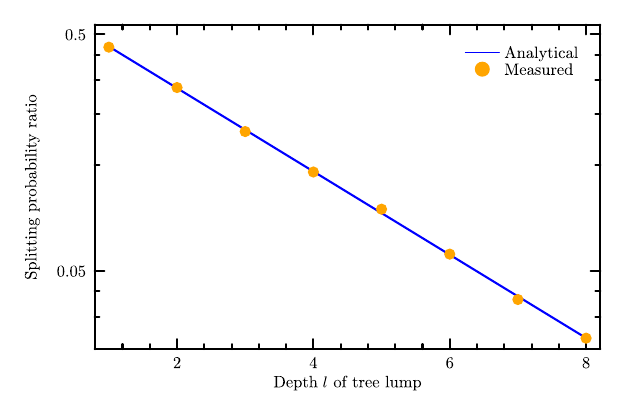}
	\caption{Splitting probability ratio for the tree lumps of depth $l\in[9]$ on a lin-log plot, tested here with $\omega=2, \; \kappa = 3$. We see that our results match with the experiments. Experiments were performed with trajectories of length $10^5$ states per tree size.}
	\label{fig:tree-affinities}
\end{figure}

\subsubsection{Splitting probability \texorpdfstring{$\Phi$}{Phi} for tree diamond lump}

We say that a tree diamond has depth $l$ if both its top and bottom tree have (edge-)~depth $l$.
Analogously to the tree lump case, we obtain for the ratio of splitting probabilities
\begin{align}
	\frac{\Phi_{+ \mid \diamondsuit, +}}{\Phi_{- \mid \diamondsuit, -}} = \left(\frac{\omega}{\kappa}\right)^{2l+1},
\end{align}
which we compare in Fig.~\ref{fig:diamond-affinities} to simulation results.

\begin{figure}[ht]
	\centering
	\includegraphics{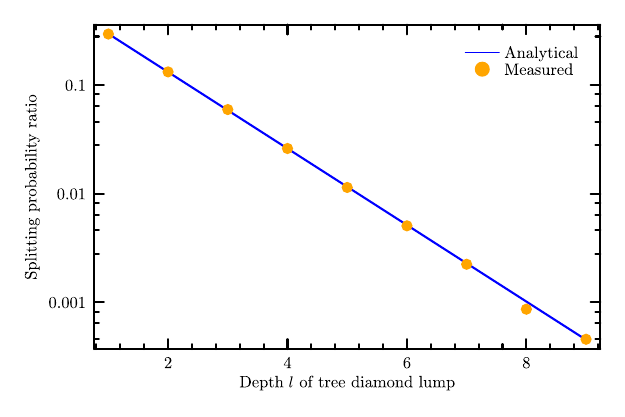}
	\caption{Splitting probability ratio for the tree diamond lump of depth $l\in[9]$ in a lin-log plot, tested here with $\omega=2, \; \kappa = 3$. The analytical results agree fully with computer experiments, which were performed with trajectories of length $5 \times 10^5$ states per tree diamond size.}
	\label{fig:diamond-affinities}
\end{figure}

\label{sec:tree-diamond-mesosystem-entropy-production-rate}
Performing a lumping of depth $l$, we have $\frac{d-l}{l+1}$ many top-up triangle ($\triangle$) levels, one tree diamond ($\diamondsuit$) level,  and $\frac{d-l}{l+1}$ many bottom-up triangle ($\triangledown$) levels. If we group each level together and describe with $J = \frac{\omega-\kappa}{n}$ the net flux between each level (which is constant across levels), we have the entropy production rate
\begin{align}
	J^{-1} \dS_2^{\rm est} &=\frac{d-l}{l+1} \ln\left(\frac{\Phi_{+ \mid \triangle, +}}{\Phi_{- \mid \triangle, -}}\right)+\frac{d-l}{l+1} \ln\left(\frac{\Phi_{+ \mid \triangledown, +}}{\Phi_{- \mid \triangledown, -}} \right) + \ln \left(\frac{\Phi_{+ \mid \diamondsuit, +}}{\Phi_{- \mid \diamondsuit, -}}  \right)\\
	&=\frac{d-l}{l+1} \ln\left(
	\left(\frac{\omega}{\kappa} \right)^{l+1}
	\right)+\frac{d-l}{l+1} \ln\left(
	\left(\frac{ \kappa}{\omega} \right)^{-(l+1)}
	\right) + \ln \left(
	\left(\frac{\omega}{\kappa}\right)^{2l+1}
	\right)\\
	&= (2d+1) \ln \left(
	\frac{\omega}{\kappa}\right),
\end{align}
which is exactly the microscopic entropy production rate, cf.\ Eq.~\eqref{eq:diamond-Sdot-micro-appendix}.

\subsection{Average lump size}
To get the average lump size, i.e.\ the average number of vertices per lump, for a given tree~diamond~graph of depth $d$ with lump depth $l$, we first find the number of tree diamond lumps and tree lumps.
For the former, we note that every diamond lump is rooted at level $d-l$, hence the number of tree~diamond~lumps is
\begin{align}
	n_{\diamondsuit} = 2^{d-l}.
\end{align}
For the latter, we observe that $2(2^{d-l}-1)$ many vertices belong to tree lumps, of which each has a size of $2^{l+1}-1$, yielding
\begin{align}
	n_{\triangle} = \frac{2(2^{d-l}-1)}{2^{l+1}-1}
\end{align}
many triangle lumps.
Thus, the average size per lump is 
\begin{align}
	\lambda_{l,d}= \frac{n_{vtcs}}{n_\diamondsuit+n_\triangle} = \frac{(2^{l+1}-1)(3\cdot 2^d -2)}{2(2^d-1+2^{d-l-1})} \stackrel{d \gg l \gg 1}{\simeq} 3 \cdot 2^l.
\end{align}

\subsection{Mistaken dynamics}\label{sec:tree-diamond-mistaken-entropy-production-rate}
We now explore what happens if the mesoscopic dynamics are \emph{mistakenly} assumed to be Markovian with effective rates. As before, we enumerate the vertex levels of the tree diamond graph by $0, \ldots, 2d$.

Since the steady-state distribution is uniform over all vertices, a lumped tree of depth $l$ has probability 
\begin{align}
	\Prob{\triangle}= \left( 2^{l+1}-1\right) \frac{1}{n}.
\end{align}
Suppose the triangle is rooted in level $U_i$, with w.l.o.g\ $i\in \{0, \ldots, d-2l-1\}$ \footnote{The last level in the upper half of the macroscopic tree diamond that roots a tree lump is level $d-2l-1$.}.
The rate from the triangle upwards, following Eq.~(4) in \cite{SUPyuInversePowerLaw2021}, is 
\begin{align}
	k_{\triangle_i \uparrow} = \frac{\kappa}{2^i \cdot  \left( 2^{l+1}-1\right)},
\end{align}
and the rate to a single lump downward is
\begin{align}
	k_{\triangle_i \downarrow} = \frac{\omega}{2^{i+l+1} \cdot  \left( 2^{l+1}-1\right)}.
\end{align}	
For $i\in \{l+1, \ldots, d-2l-1\}$ and an edge between two tree lumps rooted at levels $U_{i-l-1}$ and $U_i$, we have
\begin{align}
	&\dot S_{\text{single edge }\triangle_{i-l-1} \leftrightarrow \triangle_{i}} \\
	=&  \Prob{\triangle} \left( k_{\triangle_{i-l-1} \downarrow} - k_{\triangle_{i} \uparrow} \right)\ln \left(\frac{k_{\triangle_{i-l-1} \downarrow}}{k_{\triangle_{i} \uparrow}}\right)\\
	=& \frac{1}{n 2^i} \left(\omega - \kappa \right) \ln \left(\frac{\omega}{\kappa}\right).
\end{align}
There are $2^i$ such edges, such that we have 
\begin{align}
	\dot S_{\text{all edges on level }\triangle_{i-l-1} \leftrightarrow \triangle_{i}} =  \frac{1}{n} \left(\omega - \kappa \right) \ln \left(\frac{\omega}{\kappa}\right).
\end{align}
For a tree diamond lump of depth $l$, we have analogously 
\begin{align}
	\Prob{\diamondsuit}= \left( 3\cdot 2^{l}-2\right) \frac{1}{n},
\end{align}
with rates for a tree diamond lump rooted at level $i=d-l$ according to \cite{SUPyuInversePowerLaw2021}
\begin{align}
	k_{\diamondsuit_i \uparrow} = \frac{\kappa}{2^i \cdot  \left( 3\cdot 2^{l}-2\right)}\\
	k_{\diamondsuit_i \downarrow} = \frac{\omega}{2^i \cdot  \left( 3\cdot 2^{l}-2\right)}.
\end{align}
Thus, for all edges from the triangles to diamonds in the upper half, we have 	
\begin{align}
	\dot S_{\text{all edges on level }\triangle_{d-2l-1} \leftrightarrow \diamondsuit_{d-l}} &=  \frac{1}{n}  \left( \omega -\kappa \right) 
	\ln \left(\frac{\omega \left(3 \cdot 2^l -2 \right)}{\kappa \left(2^{l+1}-1\right)}\right),\\
	\dot S_{\text{all edges on level }\diamondsuit_{d-l} \leftrightarrow \triangle_{d+l+1}  } &=  \frac{1}{n}  \left( \omega -\kappa \right) 
	\ln \left(\frac{\omega \left(2^{l+1}-1\right) }{\kappa \left(3 \cdot 2^l -2 \right)}\right).
\end{align}
Finally, for the single edge connecting to the next macroscopic diamond tree, we have 
\begin{align}
	\dot S_{\text{macroscopic edge}} =  \frac{1}{n} ( \omega - \kappa) \ln \left(\frac{\omega}{\kappa}\right).
\end{align}	
We have $2 \cdot \frac{d-l}{l+1}-2$ many edge-levels between triangle levels, $2$ edge-levels between triangles and tree-diamond lumps, and one macroscopic edge, leading to an overall \emph{mistaken entropy production rate} 
\begin{align}
	\dot S_M &= \frac{1}{n(d)} \left(\omega - \kappa \right) \left( 2 \cdot \frac{d-l}{l+1}+1 \right)  \ln \left(\frac{\omega}{\kappa}\right) \neq \dot S = \dS_2^{\rm est},
\end{align}
which is reported in the Letter.

\section{Interpreting scalings under consistent time reversal}
\label{sec:SI:interpreting-scalings}
An apparent power-law dependence of $\dSM$ on the coarse-graining scale
$\lambda$ generally has \emph{no} relation with the
microscopic dissipation mechanisms. However, meaningful information about the underlying dynamics may still be inferred from coarse-grained data under
consistent time reversal. More specifically, if we can find an order $n$ such that $\dS^{\rm est}_{n+k}=\dS^{\rm est}_{n}$ for any $k\ge 1$. If for such $n,k$ we find that $\dS^{\rm est}_{n+k}\propto \lambda^{-\alpha}$ for some $\alpha>0$, one \emph{concludes} that the coarse graining hides dissipative cycles (this is \emph{not} true for a scaling of $\dSM$). If in addition, we have information supporting a self-similar network structure, we may also assume that most energy is dissipated on the smallest and less energy on larger scales.

It is important to emphasize that generally \emph{no} such conclusion can be drawn if we observe a coarse-graining scale independence, i.e.\ ${\dS^{\rm
		est}_{n+k}\propto\lambda^0}$ for $\lambda>1$. Only if we observe independence extending to $\lambda=1$ we may conclude that no dissipative cycles become
hidden by the coarse graining. Extrapolations $\lambda\to 1$ generally yield erroneous conclusions about the microscopic dissipation time scales (for a  
detailed discussion, see Appendix~B).

\section{6-state model}
In the 6-state model depicted in Fig.~1a in the Letter, transitions marked with a $\rightarrow$ happen at rate $\omega$ in arrow-direction, and at rate $\kappa$ in opposite direction. As indicated in the figure, along the curved arrow from state $v_3$ to state $v_1$, the transition happens at rate $6 \omega$ in arrow-direction, and at rate $\kappa$ in opposite direction. Transition along the $\leftrightarrow$ between $v_3$ and $v_4$ happens at rate $1$ in both directions. For details on the parameter choice see Section~\ref{sec:choice-of-params}.

\section{Sufficient condition for semi-Markov process having at most order \texorpdfstring{$2$}{2}}

If in the lumping every pair of lumps is connected by at most one microscopic edge, the lumped dynamics are at most $2^{\rm nd}$-order semi-Markovian. 

This holds due to the following observation: Suppose vertex $v$ is in lump $L_v$, and vertex $u$ is in an adjacent lump $L_u$, connected by the microscopic edge $e=(v,u)$. Suppose the lumped dynamics transitions at time $t$ from $L_v$ to $L_u$. From knowing this sequence of two lumps and due to uniqueness of $e$ (assumption), we thus know that at time $t$, the current microscopic state is $u$. Hence, the memory is of order (at most) $2$. Further, since transitioning from lump $L_u$ to the next lump may lead via other vertices in $L_u$, the dynamics may be \emph{semi}-Markovian.

We can observe this in the Cycle graph (Example~1), the Tree graph (Example~2), as well as the Sierpinski-type graph (Example~3). The condition is, however, not met for the Brusselator (Example~4).

We thank the Reviewers for their constructive observation on this.

\section{Notion of self-similarity}
Self-similarity is described by Ref.~\cite{SUPfalconerFractals} as ``each set comprises several smaller similar copies of itself''. From this notion, multiple definitions of self-similarity of graphs have evolved \cite{SUPsongSelfsimilarityComplexNetworks2005, SUPkronGrowthSelfSimilarGraphs2004, SUPmalozemovPurePointSpectrum1995,SUPmalozemovSelfSimilarityOperatorsDynamics2003}. The definition in \cite{SUPsongSelfsimilarityComplexNetworks2005} focuses on an extension of the box-covering idea: For a given graph $G=(V,E)$ and integer $l \in \N$, we cover $G$ by subgraphs $G_1, \ldots, G_N$ of diameter (with respect to the edge distance) \emph{strictly less than} $l$ such that $V(G) = V(G_1) \cup \ldots \cup V(G_N)$. Let $N_l(G)$ denote the minimal such $N$. We have self-similarity by the box-counting method if there is a power law dependence between $N_l$ and $l$. That is, if there exists a $d\in \N$ such that for any reasonable $l$ we have 
$$ N_l \propto l^{-d}$$
the graph is said to have fractal dimension $d$ by the box counting method \cite{SUPsongSelfsimilarityComplexNetworks2005}.

The cycle graph is with this regard a trivial example with $ N_l \propto l^{-1}$, i.e.\ dimension $d=1$. The brusselator and its graphical representation is chosen as in \cite{SUPyuInversePowerLaw2021}, where the box-counting notion as described above is used to describe self-similarity. We refer to the discussion of self-similarity therein. 
For the Sierpinski-style graph the self-similarity is evident from the following observation: Starting with a simple triangle graph, say $G^{(0)}$, we add one vertex on each edge to obtain a cycle graph (isomorphic to $C_6$) and add six further inner vertices to obtain a graph as depicted in the inset of Fig.~5a in the main text. Call this graph $G^{(1)}$. To obtain an even finer $G^{(2)}$, we apply this construction recursively on each of the four triangle-inducing subgraphs of $G^{(1)}$. Thus, the graph is by construction self-similar, similar to the constructions in \cite{SUPmalozemovSelfSimilarityOperatorsDynamics2003,SUPkronGrowthSelfSimilarGraphs2004}. An analogous construction follows for the tree graph, which consists of many copies of smaller trees. Note that in the case of the tree graph, we added an edge from each leaf to the root, to make the graph physically interesting (this allows for non-zero entropy production in the steady state).

\section{Comment on Scaling Law under additional time-scale separation Assumption}

In Appendix~C of the letter, we perform the Gedanken-Experiment of assuming in addition to Ref.~\cite{SUPyuInversePowerLaw2021} time-scale separation. This rectifies the Markovian assumption for scales of the time-scale separation and smaller. We have shown that a scaling of the (observed) entropy production rate is impossible. Here, we give the weights of the associated numerical experiment in Fig.~\ref{fig:c-36}.

\definecolor{pinkCustom}{HTML}{E80CC1}
\definecolor{intenseBlue}{HTML}{00D6CB}
\definecolor{yellowCustom}{HTML}{EBAF06}

\newcommand{\ColorOf}[1]{%
	\pgfmathsetmacro{\val}{#1}%
	\ifdim \val pt > 9999pt
	\def\edgecol{intenseBlue}%
	\else
	\ifdim \val pt > 999pt
	\ifdim \val pt < 10000pt
	\def\edgecol{pinkCustom}%
	\else
	\def\edgecol{black}%
	\fi
	\else
	\def\edgecol{yellowCustom}%
	\fi
	\fi
}

\newcommand{\CycleGraph}[3]{%
	\def\n{#1}%
	\def\forwardWeights{#2}%
	\def\backwardWeights{#3}%
	\begin{tikzpicture}[>=Stealth, every node/.style={font=\small}]
		\def\radius{8cm} 
		
		\foreach \i in {0,...,\numexpr\n-1\relax}{
			\node[circle, draw, minimum size=5mm, inner sep=0pt,fill=gray!30]
			(v\i) at ({360/\n * \i+90}:\radius) {\the\numexpr\i+1\relax};
		}
		\foreach \i in {0,...,\numexpr\n-1\relax}{
			\pgfmathtruncatemacro{\j}{mod(\i+1,\n)} 
			
			\pgfmathparse{\forwardWeights[\i]}\let\fw\pgfmathresult
			\pgfmathparse{\backwardWeights[\i]}\let\bw\pgfmathresult
			
			\ColorOf{\fw}
			\draw[->, bend left=48, draw=\edgecol]
			(v\i) to node[midway, fill=white, inner sep=1pt, text=\edgecol] {\fw} (v\j);
			
			\ColorOf{\bw}
			\draw[->, bend left=48, draw=\edgecol]
			(v\j) to node[midway, fill=white, inner sep=1pt, text=\edgecol] {\bw} (v\i);
		}
	\end{tikzpicture}%
}

\begin{figure}
	\centering
	\CycleGraph{36}{{
			1001, 1001, 10001, 1001, 1001, 2, 1001, 1001, 10001, 1001, 1001, 2, 1001, 1001, 10001, 1001, 1001, 2, 1001, 1001, 10001, 1001, 1001, 2, 1001, 1001, 10001, 1001, 1001, 2, 1001, 1001, 10001, 1001, 1001, 2
	}}{{
			1000, 1000, 10000, 1000, 1000, 1, 1000, 1000, 10000, 1000, 1000, 1, 1000, 1000, 10000, 1000, 1000, 1, 1000, 1000, 10000, 1000, 1000, 1, 1000, 1000, 10000, 1000, 1000, 1, 1000, 1000, 10000, 1000, 1000, 1
	}}
	\caption{
		Cycle graph $C_{36}$ with directional weights: We have slow edges (yellow), fast edges (pink), and very fast edges (blue). At the lumping of size $6$ (first lump are vertices $1, \ldots , 6$, etc.), we observe a time scale separation.}
	\label{fig:c-36}
\end{figure}

\section{Notion of dissipative cycles}

We want to highlight the difference between a cycle and a \emph{dissipative} cycle: \emph{Not} every cycle is a dissipative cycle. For instance, in Fig.~\ref{fig:tree-diamond}, cycles get hidden, but the dissipative ones actually do \emph{not}.
The term ``dissipative cycle getting hidden'' means that a dissipative cycle would disappear entirely inside a lump or become indistinguishable from others.\vspace{0.2cm}\\ 
We can determine the set of dissipative cycles via a spanning tree \cite{SUPschnakenbergNetworkTheoryMicroscopic1976}: Let our graph be $G$. We draw below a spanning tree $T$ for the graph of Fig.~\ref{fig:tree-diamond} in green. The set of \emph{all} dissipative cycles is spanned by the set of cycles created when closing a single edge $e \in E(G)\setminus E(T)$. We depict in Fig.~\ref{fig:dissipative-cycle} one such cycle in orange (middle figure) and pink (right figure). Those are the cycles indicated in the inset of Fig.~\ref{fig:tree-diamond}.

\begin{figure}[hbt]
	\includegraphics{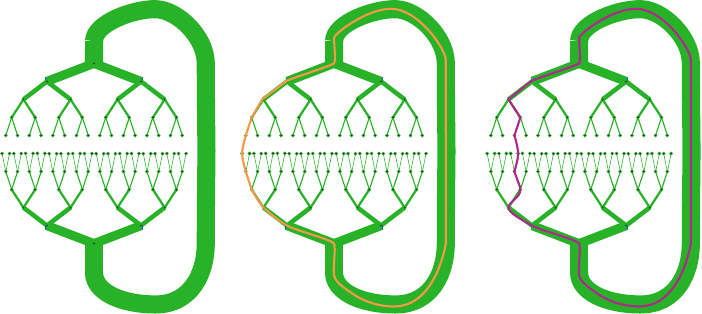}
	\caption{
		The spanning tree (left) of the tree-diamond graph and two dissipative cycles obtained via the spanning tree.
	}
	\label{fig:dissipative-cycle}
\end{figure}
In the diamond-tree graph, no lumping (except for the trivial lumping, where all vertices form one single big lump) will hide any dissipative cycles. For every lumping and every spanning cycle, at least one edge of that spanning cycle will transition between two distinct lumps. Indeed, spanning cycles which are too close to be distinguished under a current lumping will appear as a single, stronger cycle, but these are \emph{not} ``hidden''; The steady state current through the two microscopic cycles will be summed to the current through the mesoscopic cycle.\vspace{0.2cm}\\ 
If there were cycles hidden by the lumping---as it is the case in the Sierpinski-style graph---then there would exist a cycle which fully disappears inside a single lump.
We can thus conclude that \emph{no} dissipative cycles get hidden in the tree-diamond graph.

\section{Technicalities of the sampling procedure}

\subsection{Exploiting the symmetry of the graph}
In the cases of the tree graph and tree diamond graph, we use the graph symmetry to increase the statistical precision: all lumps on a given level have (1) equal steady state probability, (2) equal transition rates, predecessors, and successors (important for the memory of the stochastic process), and (3) no edge connects them. Thus, after simulating the microscopic dynamics, we may in the statistical analysis w.l.o.g.\ treat all lumps on a given level as equal, which increases the statistics per lump. Note that this is a very particular property of the problems' symmetry and may \emph{not} be used in general, e.g.\ not for the Brusselator.
\subsection{Choice of parameters}\label{sec:choice-of-params}
For the results shown in the Letter, we chose $w=2$, $\kappa=1$ for both the ring and tree graph. For the tree diamond graph, we chose $w=2, \kappa =3$.

For the Sierpinski-type graph (Fig.~5c in the Letter), trajectories were drawn after $N_{\rm skip}=10^9$ initial steps and are of length $N_{\rm steps}=5 \times 10^8$. The plot shows the average of $108$ trajectories, with almost vanishing error bars. We took $\omega =2$ for the driving.

For the Brusselator, (Fig.~5d in the Letter), trajectories were drawn after $N_{\rm skip}=10^9$ initial steps and are of length $N_{\rm steps}=10^9$. The plot shows the average of $25$ trajectories with almost vanishing error bars. We note that the state space is very large with many states having steady-state probability almost $0$ but non-vanishing transition rates. In the microscopic setting, $\lambda=1$, this may lead to undersampling of the steady state.
Similarly, order $k\geq 4$ would need even longer trajectories.

For the tree-diamond graph (Fig.~\ref{fig:tree-diamond}b), we used $\omega=2$, $\kappa =3$, a depth of $d=23$, and the trajectories were drawn after $N_{\rm skip}=10^8$  and recorded for $N_{\rm steps}=10^9$ steps, repeated 10 times.

In all cases, the fitted lines were least-squares-fitted to the measurements on all scales except for the smallest and largest scale \cite{SUPyuInversePowerLaw2021}, which clearly deviate in the tree, diamond tree, Sierpinski-type graph, and Brusselator, i.e.\ where we do \emph{not} have a power law of the entropy production rate in the scale.

For the 6-state-model (Fig.~1 and Appendix~A in the Letter), we chose $\omega =2, \kappa = 3$ and ran for each line 30 simulations. For equilibration, we used $N_{skip}= 10^6$ steps. Note that those steps indicated here are always the \emph{microscopic} steps, which then get coarse grained. Thus, the estimator sees effectively less \emph{mesoscopic} transitions.

\subsection{Subsampling for higher orders}

For higher-order estimators, we need exponentially longer trajectories, since statistics of tuples of length $k+1$ for order $k$ estimators need to be drawn and sampled sufficiently. We see in Fig.~\ref{fig:undersampling}a that the order $k=4$ estimator struggles for the coarsest scales, however \emph{not} due to high fluctuations among the $100$ trajectory measurements, but instead as a systematic under-sampling error. For $5$ times longer trajectories (Fig.~\ref{fig:undersampling}b) we see significant enhancement on the second coarsest scale. Increasing the trajectory length by a factor of $10$ (Fig.~\ref{fig:undersampling}c) does not yield significant improvement, underlining the non-linear requirement of sample statistics in the order $k$.
\begin{figure}[hbt]
	\includegraphics{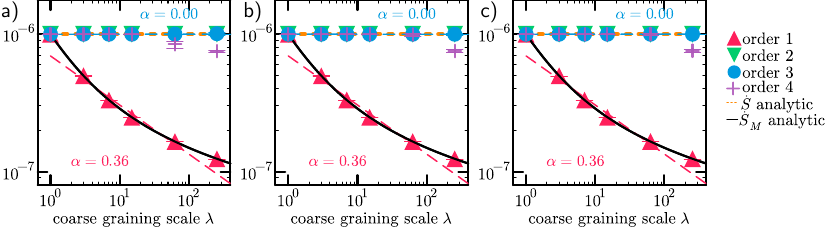}
	\caption{Entropy production rate on the tree graph as a function of the coarse-graining scale for different sample path lengths: (a) For paths of length $10^8$ (b) of length $5\cdot 10^8$ (c) of length $10^9$. All trajectories were recorded after initial $N_{\rm skip}=10^8$ many steps. Depicted are averages of $100$ such trajectories with standard deviation indicated through the error bars.}
	\label{fig:undersampling}
\end{figure}

\end{document}